\documentclass[journal]{IEEEtran}
\usepackage{cite}
\usepackage{amsmath,amssymb,amsfonts}
\usepackage{amsthm}
\usepackage{algorithmic}
\usepackage{graphicx}
\usepackage[caption=false,font=footnotesize]{subfig}
\usepackage{textcomp}
\usepackage{xcolor}
\usepackage[scaled]{helvet} 
\usepackage{balance}
\usepackage{makecell}
\usepackage{makecell}
\usepackage{xcolor} 

\newtheorem{remark}{Remark}

\newtheorem{lemma}{\bfseries Lemma}

\def\BibTeX{{\rm B\kern-.05em{\sc i\kern-.025em b}\kern-.08em
    T\kern-.1667em\lower.7ex\hbox{E}\kern-.125emX}}
\begin{document}




\title{
    Asymptotically Optimal Local Receiver in Uplink CF-mMIMO: A Functional-Variational Analysis}
  \author{Jiafei Fu,~\IEEEmembership{Graduate Student Member,~IEEE}
  , Peng Jiang,~\IEEEmembership{Graduate Student Member,~IEEE}, \\Dongming Wang,~\IEEEmembership{Member,~IEEE},
  Pengcheng Zhu,~\IEEEmembership{Member,~IEEE}, 
\thanks{This work was supported by Mobile Information Networks - National Science and Technology Major Projects of China under Grant 2025ZD1301800
and by the National Natural Science Foundation of China under Grant
62531003.}
\thanks{Jiafei Fu, Peng Jiang, Dongming Wang, and Pengcheng Zhu are with National Mobile Communications Research Laboratory, Southeast University, Nanjing 210096,
China. (e-mails: 
\{fujfei, jiangpeng98, wangdm, p.zhu\}@seu.edu.cn). Corresponding author: Pengcheng Zhu.}
}

\maketitle

\begin{abstract}
In cell-free massive multiple-input multiple-output (CF-mMIMO) systems, the canonical uplink local receiver is the local minimum mean square error (LMMSE) receiver with large-scale fading decoding (LSFD) at the central processing unit (CPU). The LSFD coefficients are derived under the use-and-then-forget (UatF) lower bound of the ergodic rate, and computing these coefficients introduces additional fronthaul overhead and computational complexity at the CPU. This paper investigates local receiver design directly from the true ergodic-rate objective under perfect local channel state information (CSI). By introducing an expectation-based constraint and leveraging large-system random matrix theory, we develop a functional-variational approach that yields the asymptotically optimal quasi-LMMSE (Q-LMMSE) receiver in closed form. A key insight is that the Q-LMMSE receiver shares the same direction as the conventional LMMSE receiver, differing only by an instantaneous CSI-dependent scalar, and thus incurs the same per-access point (AP) complexity. More importantly, this scalar varies across APs and implicitly provides adaptive weighting for the direct summation at the CPU, thereby completely eliminating the need for statistical LSFD coefficients and the associated CPU-side computational overhead. Numerical results demonstrate that the proposed Q-LMMSE receiver consistently outperforms the LMMSE-LSFD benchmark in terms of the ergodic rate, achieving approximately a {5\%} gain when the number of antennas per AP is low, while operating with strictly lower system-level complexity.
\end{abstract}


\section{Introduction}
{
 The relentless growth of mobile data traffic and the emergence of new application scenarios---such as ultra-reliable low-latency communications, massive machine-type communications, and enhanced mobile broadband---have driven the wireless community to seek fundamentally new network architectures beyond the traditional cellular paradigm~\cite{Michail2021CM64,You2020SCIS,chengxiang2023CST905}. 
 However, in conventional cellular massive multiple-input multiple-output (MIMO), cell-edge users still suffer from pronounced inter-cell interference and degraded signal quality due to their distance from the serving base station~\cite{Peng2026tsp2224}

Cell-free massive MIMO (CF-mMIMO) has emerged as a promising user-centric architecture to address this limitation~\cite{NGO2017TWC,Ammar2022COMST}. In a CF-mMIMO system, a large number of geographically distributed access points (APs) jointly serve a smaller set of user equipment (UEs) over the same time-frequency resources, ensuring that every user is surrounded by several nearby APs and thus experiences favorable channel conditions regardless of its physical location~\cite{NGO2017TWC}. 
By leveraging macro-diversity across distributed APs and coherent joint processing, CF-mMIMO can substantially improve coverage uniformity, spectral efficiency, and interference resilience compared with conventional small-cell deployments~\cite{NGO2017TWC,DEMIR2024JSAC}, positioning it as a key enabling technology for beyond-5G and 6G wireless systems~\cite{ning2025chinaCommunications1,zhang2024INETWORK247,Elhoushy2022csto492}. 

{However, this architectural shift from co-located to distributed antenna deployment introduces fundamental challenges for uplink signal processing. While a fully centralized implementation---where all APs forward raw received signals to a central processing unit (CPU) for joint detection---can theoretically achieve optimal performance, it imposes prohibitive fronthaul overhead and computational complexity. To ensure scalability, practical CF-mMIMO systems typically adopt a distributed processing paradigm, where each AP performs local signal detection using only its locally available channel state information (CSI), thereby significantly reducing fronthaul and computational requirements~\cite{Jiayi2026tcom8459}.}

The dominant analytical framework for local receiver design in CF-mMIMO has been the use-and-then-forget (UatF) bounding technique~\cite{NGO2017TWC,OZDOGAN2019TWC}. Within this framework, the local minimum mean squared error (LMMSE) receiver followed by large-scale fading decoding (LSFD) at the CPU---often referred to as LMMSE-LSFD---has emerged as the canonical receiver structure~\cite{BJORNSON2021Book,Yu2026tvt1671}. Due to its analytical tractability and computational efficiency, the LMMSE-LSFD architecture has been extensively studied and successfully extended to a wide range of practical scenarios, including fronthaul-constrained deployments, Rician fading, multi-antenna users, and hardware impairments~\cite{DEMIR2024JSAC,OZDOGAN2019TWC,XIE2024TWC,Tentu2023tcom5455}.

In parallel, recent advances have further extended the uplink combining design to diverse practical scenarios. On the algorithmic front, iteratively weighted MMSE based uplink precoding and combining designs have been developed to improve spectral efficiency~\cite{Wang2023TCOM}, and optimal bilinear equalizer beamforming has been proposed for both uplink and downlink CF-mMIMO with arbitrary channel estimators~\cite{Wang2025TVT}. Beyond algorithmic refinements, emerging system architectures have introduced new combining challenges and opportunities. For instance, distributed combining schemes for near-field cell-free extremely large-scale MIMO systems have been investigated to reduce computational complexity while maintaining performance~\cite{Wang2026TWC_NearField}. The integration of CF-mMIMO with reconfigurable intelligent surfaces (RIS) has led to novel combining designs, including reflection pattern modulation-aided RIS-assisted schemes~\cite{Sui2024TCOM} and electromagnetic interference-aware receivers~\cite{Shi2023JSAC}. Dynamic time division duplexing-enabled cell-free systems with joint sensing and communication have further expanded the scope of uplink receiver design~\cite{Chowdhury2025TWC}. Complementing these architectural innovations, uplink resource allocation optimization for user-centric CF-mMIMO networks has been studied to balance performance and scalability~\cite{Li2024TWC}. 

However, despite the richness of these designs and the dominance of the UatF-based approach, the framework itself suffers from several fundamental limitations. 
First, the UatF lower bound is a conservative surrogate of the true ergodic achievable rate; the gap between the two can be non-negligible, particularly in interference-limited regimes or when channel hardening is not fully realized~\cite{NGO2017TWC,BJORNSON2017TIT}. 
{Second, the LSFD stage requires the CPU to collect local estimates from all APs and compute optimal combining coefficients based on long-term channel statistics, which introduces additional fronthaul overhead and computational burden~\cite{VanChien2021TWC,Bjornson2020TCOM}.} Third, although the LMMSE-LSFD structure has been widely adopted as the de facto receiver under the UatF framework, a rigorous optimality proof---establishing that LMMSE-LSFD is indeed the optimal local combiner under the UatF objective---has, to the best of our knowledge, not been reported. Prior works have largely taken the optimality of local MMSE combining as given, without providing a formal analysis that confirms this choice.

These observations motivate the questions addressed in this paper: \emph{Can we design local receivers directly from the ergodic-rate objective that go beyond the UatF surrogate? {And can we rigorously prove the optimality of the UatF-based receiver?}} {Actually, direct optimization of the ergodic rate has a rich history in multiuser MIMO literature, largely enabled by advances in random matrix theory and deterministic-equivalent analysis{\cite{couillet2011CUP,Sanguinetti2022FnT}}. These mathematical frameworks provide powerful tools to approximate random functionals of large-dimensional channel matrices by deterministic quantities that depend only on channel statistics~\cite{ZHANG2013TWC1536,Jiafei2026twc12548,Jiangpeng2025TWC}. 
Leveraging these techniques, researchers have successfully analyzed the asymptotic performance of linear precoding~\cite{Hoydis2013JSAC,Wagner2012TIT4509}, established the fundamental capacity limits of massive MIMO in the large-system regime~\cite{Marzetta2010TWC,BJORNSON2017TIT}, and developed energy-efficient transmission designs~\cite{You2020TWC,Bjornson2015TWC} as well as advanced beamforming schemes exploiting statistical CSI~\cite{Li2016TCOMM}.
Despite these advances in transmit-side optimization and co-located massive MIMO systems, the direct design of \emph{local receivers} from the ergodic-rate objective in CF-mMIMO has received comparatively little attention.} 

\subsection{Contributions and Paper Organization}

This paper addresses the above gap by designing uplink local receivers directly from the true ergodic-rate objective under the assumption of perfect local CSI at each AP. 
{By introducing an expectation-based surrogate constraint to replace the pointwise normalization in the instantaneous signal-to-interference-plus-noise ratio (SINR), we transform the original intractable problem into a large-system asymptotic formulation.} Through a rigorous variational analysis using functional derivatives, we derive the first-order stationarity conditions and obtain a quasi-LMMSE (Q-LMMSE) receiver with a deterministic-equivalent scalar characterization. {The cross-AP coupling terms, which capture the interference from other APs, are shown to vanish in the conditional mean under the large-system limit.} An important practical aspect of the proposed design is its computational efficiency. By exploiting the matrix inversion lemma, the proposed receiver can be computed from the standard LMMSE inverse with a instantaneous CSI-based scalar, maintaining the same per-AP complexity as the conventional LMMSE receiver while completely avoiding LSFD combination at the CPU. This is a significant advantage over the LMMSE-LSFD architecture, which requires additional CPU-side processing and additional fronthaul signaling.



The main contributions of this paper are summarized as follows:

\begin{itemize}
    \item We develop a functional-variational framework for direct ergodic-rate-oriented receiver design in CF-mMIMO. By introducing an expectation-based surrogate constraint and leveraging large-system random matrix approximations, we derive the first-order stationarity conditions and obtain an asymptotically optimal Q-MMSE receiver under the ergodic rate objective.
    \item {Using functional-variational analysis, we also formally establish that local MMSE combining followed by LSFD at the CPU is indeed the optimal receiver structure when the UatF bound is adopted as the performance metric.}
    \item {Via the matrix inversion lemma, we show that the Q-LMMSE receiver incurs the same per-AP computational complexity as standard LMMSE receiver while completely avoiding LSFD combining in the CPU, therefore achieving significant computational savings.}
    \item Numerical results across various parameter configurations demonstrate that the proposed Q-LMMSE receiver consistently outperforms the LMMSE-LSFD benchmark in terms of the ergodic rate.
\end{itemize}

The remainder of this paper is organized as follows. Section~\ref{sec:sysModel} introduces the system model. Section~\ref{sec:prelim} reviews the conventional uplink receiver in CF-mMIMO systems. Section~\ref{sec:opt_Elog} formulates the ergodic-rate-oriented receiver design problem, develops the large-system approximation and a functional-variational-based solution.
Section~\ref{sec:simu} presents numerical results, and Section~\ref{sec:conclusion} concludes the paper.

{Notations}: $\mathbb{C}$ denotes the set of complex numbers. Bold lowercase and uppercase letters denote vectors and matrices, respectively. $\mathbf{I}_N$ is the $N\times N$ identity matrix, $\mathbf{0}$ is the all-zero vector, and $\mathrm{diag}(\cdot)$ constructs a block-diagonal matrix. $\mathbf{A}^H$ and $\mathbf{A}^{-1}$ denote the conjugate transpose and inverse of $\mathbf{A}$, respectively. $\|\cdot\|$ is the Euclidean norm while $\mathcal{CN}(\boldsymbol{\mu}, \mathbf{\Sigma})$ denotes a circularly symmetric complex Gaussian distribution with mean $\boldsymbol{\mu}$ and covariance $\mathbf{\Sigma}$. $\mathbb{E}[\cdot]$ is the expectation operator, $\mathcal{R}\{\cdot\}$ is the real part operator, and $p(\cdot)$ denotes a probability density function. $\mathrm{vec}(\mathbf{A})$ denotes the vectorization of $\mathbf{A}$. $\mathcal{O}(\cdot)$ denotes the big-O notation for asymptotic complexity.

\section{System Model}\label{sec:sysModel}
We consider the uplink of a CF-mMIMO system comprising $M$ APs and $K$ single-antenna UEs. Each AP is equipped with $N$ antennas, and all APs are connected to a CPU via fronthaul links. Let $\mathcal{M} = \{1, 2, \ldots, M\}$ and $\mathcal{K} = \{1, 2, \ldots, K\}$ denote the sets of APs and UEs, respectively. The channel between UE $k$ and AP $m$ is denoted by $\mathbf{h}_{km} \in \mathbb{C}^{N}$, which is assumed to be perfectly known at AP $m$ and follows a complex Gaussian distribution:
\begin{equation}
\mathbf{h}_{km} \sim \mathcal{CN}(\mathbf{0}, \mathbf{R}_{km}),
\end{equation}
where $\mathbf{R}_{km} \in \mathbb{C}^{N \times N}$ denotes the spatial correlation matrix. We assume that the channels are independent across different UEs and APs. Let us define the aggregate channel matrix for AP $m$ as 
\begin{equation}
    \mathbf{H}_m = [\mathbf{h}_{1m}, \mathbf{h}_{2m}, \ldots, \mathbf{h}_{Km}] \in \mathbb{C}^{N \times K}, 
\end{equation}
which collects the channels from all UEs to AP $m$. Since $\mathbf{H}_m$ is a Gaussian random matrix with independent columns, its probability density function (PDF) can be expressed as
\begin{equation}    
p(\mathbf{H}_m) = \frac{1}{\pi^{NK} \prod_{k=1}^K \det(\mathbf{R}_{km})} \exp\left(-\sum_{k=1}^K \mathbf{h}_{km}^H \mathbf{R}_{km}^{-1} \mathbf{h}_{km}\right).
\end{equation}
By vectorizing $\mathbf{H}_m$ as $\mathrm{vec}(\mathbf{H}_m)= [\mathbf{h}_{1m}^T, \mathbf{h}_{2m}^T, \ldots, \mathbf{h}_{Km}^T]^T \in \mathbb{C}^{NK}$, the PDF can be equivalently written in compact form as
\begin{equation}    
    \begin{aligned}
        p(\mathrm{vec}(\mathbf{H}_m)) = &\frac{1}{\pi^{NK} \prod_{k=1}^K \det(\mathbf{R}_{km})} \\
        &\times \exp\left(-\mathrm{vec}(\mathbf{H}_m)^H \mathbf{R}_m^{-1} \mathrm{vec}(\mathbf{H}_m)\right),
    \end{aligned}
\end{equation}
where $\mathbf{R}_m = \mathrm{diag}(\mathbf{R}_{1m}, \mathbf{R}_{2m}, \ldots, \mathbf{R}_{Km}) \in \mathbb{C}^{NK \times NK}$ denotes the block-diagonal covariance matrix of the vectorized channel. 

The received signal at AP $m$ can be expressed as
\begin{equation}
\mathbf{y}_m = \sum_{k=1}^{K} \mathbf{h}_{km} \sqrt{p_k} s_k + \mathbf{n}_m,
\end{equation}
where $p_k$ denotes the transmit power of UE $k$, $s_k$ is the transmitted symbol satisfying $\mathbb{E}[|s_k|^2] = 1$, and $\mathbf{n}_m \sim \mathcal{CN}(\mathbf{0}, \sigma^2 \mathbf{I}_N)$ represents the additive white Gaussian noise (AWGN) at AP $m$. 

Each AP $m$ applies a local receive vector $\mathbf{v}_{km} \in \mathbb{C}^{N}$ to obtain a local estimate of the symbol transmitted by UE $k$, denoted as
\begin{equation}
\hat{s}_{km} = \mathbf{v}_{km}^H \mathbf{y}_m.
\end{equation}

The CPU then combines the local estimates from all APs to produce the final decision statistic for UE $k$.
The local receive vector $\mathbf{v}_{km}$ is designed by AP $m$ based on its locally available channel state information $\mathbf{H}_m$ and a set of statistical parameters, formulated as
\begin{equation}
\mathbf{v}_{km} = \mathbf{f}_{km}(\mathbf{H}_m, \mathcal{S}),
\end{equation}
where the local receiving function $\mathbf{f}_{km}(\cdot)$ satisfies the following assumptions:
\begin{itemize}
    \item[\textit{A1)}] $\mathcal{S}$ denotes a set of statistical information that is common to all APs, including $\{\mathbf{R}_{jm}\}_{j\in\mathcal{K}}$, the noise variance $\sigma^2$, and the power allocation coefficients $\{p_j\}_{j\in\mathcal{K}}$. These quantities are independent of the instantaneous channel realizations.
    \item[\textit{A2)}] The local receive vector $\mathbf{v}_{km}$ depends only on the local channel $\mathbf{H}_m$ and is statistically independent of the channels and receiving vectors at other APs, i.e., $\mathbf{v}_{km} \perp\!\!\!\perp \{\mathbf{H}_l, \mathbf{v}_{jl}\}$ for all $l \neq m$, $j \in \mathcal{K}$.
    \item[\textit{A3)}] The function $\mathbf{f}_{km}(\cdot)$ maps the local channel matrix to a combining vector, i.e., $\mathbf{f}_{km}: \mathbb{C}^{N\times K} \rightarrow \mathbb{C}^{N}$, and belongs to the Hilbert space $L^2(\mathbb{C}^{N\times K}, p(\mathbf{H}_m); \mathbb{C}^N)$. This implies square-integrability with respect to the probability measure: $\int_{\mathbb{C}^{N\times K}} \|\mathbf{f}_{km}(\mathbf{H}_m)\|^2 p(\mathbf{H}_m) \, d\mathbf{H}_m < \infty$.
\end{itemize}


The design of the local receiving functions $\{\mathbf{f}_{km}(\cdot)\}$ constitutes the main focus of this paper, as they directly determine the ergodic achievable rates. We assume perfect local CSI at each AP, which enables us to characterize the fundamental performance limits of local receiver design. The insights derived from this analysis can also guide the development of robust receivers under imperfect CSI in practical scenarios.

{

\section{{Preliminaries on Uplink Receiver in CF-mMIMO Systems}}\label{sec:prelim}

In this section, we review the conventional uplink receiver architectures in CF-mMIMO systems, starting from the theoretically optimal centralized approach and then discussing the practical distributed paradigm. These preliminaries provide the foundation for understanding the motivation behind the proposed ergodic-rate-oriented local receiver design.

\subsection{Centralized Uplink Reception}

In a fully centralized uplink reception architecture, all APs forward their raw received signals to the CPU, which performs joint signal detection using the global CSI. Let us define the aggregate received signal vector across all APs as
\begin{equation}
    \mathbf{y} = [\mathbf{y}_1^T, \mathbf{y}_2^T, \ldots, \mathbf{y}_M^T]^T \in \mathbb{C}^{MN},
\end{equation}
and the aggregate channel matrix from all UEs to all APs as
\begin{equation}
    \mathbf{H} = [\mathbf{H}_1^T, \mathbf{H}_2^T, \ldots, \mathbf{H}_M^T]^T \in \mathbb{C}^{MN \times K}.
\end{equation}
The global received signal model can then be written as
\begin{equation}
    \mathbf{y} = \mathbf{H} \mathbf{P}^{1/2} \mathbf{s} + \mathbf{n},
\end{equation}
where $\mathbf{P} = \mathrm{diag}(p_1, \ldots, p_K)$ is the diagonal power allocation matrix, $\mathbf{s} = [s_1, \ldots, s_K]^T$ is the transmitted symbol vector, and $\mathbf{n} = [\mathbf{n}_1^T, \ldots, \mathbf{n}_M^T]^T \sim \mathcal{CN}(\mathbf{0}, \sigma^2 \mathbf{I}_{MN})$ is the aggregate noise vector.

The CPU applies a global receive combining matrix $\mathbf{v}_k \in \mathbb{C}^{MN}$ to detect the symbol of UE $k$, yielding the decision statistic
\begin{equation}
    \hat{s}_k = \mathbf{v}_k^H \mathbf{y}.
\end{equation}
The instantaneous SINR for UE $k$ under centralized reception is given by
\begin{equation}\label{sinr_centralized}
    \mathrm{SINR}_k^{\mathrm{C}} = \frac{p_k \left| \mathbf{v}_k^H \mathbf{h}_k \right|^2}{\sum_{i \neq k} p_i \left| \mathbf{v}_k^H \mathbf{h}_i \right|^2 + \sigma^2 \|\mathbf{v}_k\|^2},
\end{equation}
where $\mathbf{h}_k = [\mathbf{h}_{k1}^T, \ldots, \mathbf{h}_{kM}^T]^T \in \mathbb{C}^{MN}$ denotes the aggregate channel vector from UE $k$ to all APs.

To maximize the instantaneous SINR \eqref{sinr_centralized}, the CPU solves the following generalized Rayleigh quotient optimization problem:
\begin{equation}
    \max_{\mathbf{v}_k \in \mathbb{C}^{MN}} \quad \frac{p_k \left| \mathbf{v}_k^H \mathbf{h}_k \right|^2}{\sum_{i \neq k} p_i \left| \mathbf{v}_k^H \mathbf{h}_i \right|^2 + \sigma^2 \|\mathbf{v}_k\|^2}.
\end{equation}
This is a standard generalized eigenvalue problem, and the globally optimal solution is the centralized MMSE (C-MMSE) receiver:
\begin{equation}\label{v_centralized}
    \mathbf{v}_k^{\mathrm{C-MMSE}} = \left( \sum_{i=1}^K p_i \mathbf{h}_i \mathbf{h}_i^H + \sigma^2 \mathbf{I}_{MN} \right)^{-1} \mathbf{h}_k.
\end{equation}
Equivalently, defining the global interference-plus-noise covariance matrix as $\mathbf{W} = \sum_{i=1}^K p_i \mathbf{h}_i \mathbf{h}_i^H + \sigma^2 \mathbf{I}_{MN} \in \mathbb{C}^{MN \times MN}$, the optimal receiver can be compactly written as $\mathbf{v}_k^{\mathrm{C-MMSE}} = \mathbf{W}^{-1} \mathbf{h}_k$.

The centralized MMSE receiver \eqref{v_centralized} achieves the maximum possible instantaneous SINR for each UE by exploiting the full global CSI $\mathbf{H}$ and performing joint detection across all APs. However, this theoretical optimality comes at a prohibitive cost in large-scale CF-mMIMO deployments:
\begin{itemize}
    \item \textbf{Fronthaul overhead}: Each AP must forward its $N$-dimensional received signal $\mathbf{y}_m$ to the CPU, resulting in a total fronthaul load of $MN$ complex scalars per channel use. For systems with hundreds of APs, this requirement far exceeds practical fronthaul capacity.
    \item \textbf{Computational complexity}: The CPU must invert an $MN \times MN$ matrix, incurring a complexity of $\mathcal{O}((MN)^3)$. This cubic scaling with the total number of antennas renders centralized processing computationally infeasible for large-scale deployments.
    \item \textbf{CSI acquisition burden}: The CPU must acquire and maintain the global channel matrix $\mathbf{H}$, which requires instantaneous CSI sharing among all APs and introduces significant signaling latency.
\end{itemize}
These fundamental scalability constraints motivate the adoption of distributed processing architectures, where each AP performs local signal detection using only its locally available CSI, thereby significantly reducing both fronthaul and computational requirements.

\subsection{Distributed Uplink Reception via LMMSE-LSFD}

The LMMSE-LSFD architecture adopts a two-stage processing paradigm that balances performance with implementation complexity.
In the first stage, each AP $m$ independently computes a local receive vector $\mathbf{v}_{km} \in \mathbb{C}^{N}$ based solely on its locally available channel $\mathbf{H}_m$, and produces a local symbol estimate: $\hat{s}_{km} = \mathbf{v}_{km}^H \mathbf{y}_m$.
In the second stage, the CPU combines the local estimates from all APs using LSFD coefficients $\mathbf{a}_k = [a_{k1}, \ldots, a_{kM}]^T \in \mathbb{C}^{M}$ to produce the final decision statistic:
\begin{equation}
    \hat{s}_k = \sum_{m=1}^{M} a_{km} \hat{s}_{km} = \mathbf{a}_k^H \hat{\mathbf{s}}_k,
\end{equation}
where $\hat{\mathbf{s}}_k = [\hat{s}_{k1}, \ldots, \hat{s}_{kM}]^T$ collects the local estimates across all APs.

Under the UatF analytical framework, the optimal local receive vector at each AP is derived by maximizing the UatF lower bound on the capacity. The resulting local receiver takes the standard LMMSE form:
\begin{equation}\label{v_lmmse_local}
    \mathbf{v}_{kl}^{\mathrm{LMMSE}} = \left( \sum_{i=1}^K p_i \mathbf{h}_{il} \mathbf{h}_{il}^H + \sigma^2 \mathbf{I}_N \right)^{-1} \mathbf{h}_{kl}, \quad \forall l \in \mathcal{M}.
\end{equation}
This local LMMSE treats all other users as interference and requires only the local channel $\mathbf{H}_l$, making it computationally efficient with per-AP complexity $\mathcal{O}(N^3 + KN^2)$.
\begin{proof}
    A detailed proof of the LMMSE optimality based on functional-variational theory is given in Appendix~\ref{app:uatf_proof}.
\end{proof}

Given the local LMMSE reveiver \eqref{v_lmmse_local}, the UatF-based SINR for UE $k$ can be expressed as a deterministic function of the LSFD coefficients:
\begin{equation}\label{sinr_uatf}
   \begin{aligned}
    &\mathrm{SINR}_k^{\mathrm{UatF}} \\&= \frac{p_k \mathbf{a}_k^H \mathbb{E}[\mathbf{g}_{kk}] \mathbb{E}[\mathbf{g}_{kk}^H] \mathbf{a}_k}{\mathbf{a}_k^H \left( \sum\limits_{i=1}^K p_i \mathbb{E}[\mathbf{g}_{ki} \mathbf{g}_{ki}^H] - p_k \mathbb{E}[\mathbf{g}_{kk}] \mathbb{E}[\mathbf{g}_{kk}^H] + \mathbf{\Gamma}_k \right) \mathbf{a}_k},
    \end{aligned}
\end{equation}
where $\mathbf{g}_{ki} = [\mathbf{v}_{k1}^H \mathbf{h}_{i1}, \ldots, \mathbf{v}_{kM}^H \mathbf{h}_{iM}]^T$ and $\mathbf{\Gamma}_k = \sigma^2 \mathrm{diag}(\mathbb{E}[\|\mathbf{v}_{k1}\|^2], \ldots, \mathbb{E}[\|\mathbf{v}_{kM}\|^2])$. Maximizing \eqref{sinr_uatf} with respect to $\mathbf{a}_k$ yields the optimal LSFD coefficients:
\begin{equation}\label{a_lsfd_opt}
    \mathbf{a}_k^{\mathrm{opt}} = \mathbf{T}_k^{-1} \mathbb{E}[\mathbf{g}_{kk}],
\end{equation}
where $\mathbf{T}_k = \sum_{i=1}^K p_i \mathbb{E}[\mathbf{g}_{ki} \mathbf{g}_{ki}^H] - p_k \mathbb{E}[\mathbf{g}_{kk}] \mathbb{E}[\mathbf{g}_{kk}^H] + \mathbf{\Gamma}_k$.
\begin{remark}
The local LMMSE receiver combined with the CPU-level LSFD coefficients forms the jointly optimal receive combining structure under the UatF metric.
\end{remark}

\subsubsection{Advantages and Limitations}

The LMMSE-LSFD architecture offers several practical advantages over centralized processing. 
In terms of fronthaul overhead, each AP forwards only one scalar $\hat{s}_{km}$ per UE to the CPU, reducing the fronthaul load from $MN$ to $MK$ complex scalars per channel use---a factor of $N/K$ reduction when $N > K$. 
Computationally, the $N \times N$ matrix inversion is performed locally at each AP with per-AP complexity $\mathcal{O}(N^3 + KN^2)$, enabling parallel processing across APs. 
Moreover, the UatF framework yields closed-form expressions for both local combiners and LSFD coefficients, facilitating system-level analysis and optimization.

However, this architecture also suffers from several fundamental limitations. 
First, the UatF lower bound replaces random channel quantities with their statistical expectations before applying the logarithm, yielding a conservative surrogate of the ergodic rate $\mathbb{E}[\log_2(1+\mathrm{SINR}^{\mathrm{INS}})]$; the gap between the two can be non-negligible, particularly in interference-limited regimes or when channel hardening is not fully realized. 
Second, the LSFD coefficients \eqref{a_lsfd_opt} depend on long-term channel statistics and cannot adapt to instantaneous channel variations, which limits the ability to exploit favorable instantaneous channel conditions. 
Third, computing the LSFD coefficients requires estimating high-dimensional statistical expectations (e.g., $\mathbb{E}[\mathbf{g}_{ki} \mathbf{g}_{ki}^H]$), which typically involves Monte Carlo simulations with $T$ samples, incurring a global complexity of $\mathcal{O}(TK^2M(M+N) + KM^2(K+M))$.

These limitations motivate the search for alternative distributed receiver designs that directly optimize the ergodic rate while avoiding the CPU-side LSFD computation overhead.


}

\section{Ergodic-Rate-Oriented Local Receiver Design}\label{sec:opt_Elog}

In this section, we aim to maximize the ergodic rate for each UE by optimizing the local receive vectors $\{\mathbf{v}_{km}\}$. Since the ergodic rate maximization problem inherently constitutes a challenging fractional programming problem with an expectation-of-logarithm structure, we propose a systematic solution framework to solve it.
%
\subsection{Problem Formulation for Ergoic Rate Maximization}
The instantaneous signal-to-interference-plus-noise ratio (SINR) for UE $k$ after combining the local estimates from all APs can be expressed as
\begin{equation}
\mathrm{SINR}_k^{\mathrm{INS}} = \frac{p_k \left|\sum_{m=1}^M \mathbf{v}_{km}^H \mathbf{h}_{km}\right|^2}{ \sum_{i \neq k} p_i\left|\sum_{m=1}^M\mathbf{v}_{km}^H \mathbf{h}_{im}\right|^2 + \sigma^2 \sum_{m=1}^M\mathbf{v}_{km}^H \mathbf{v}_{km}},
\end{equation}
where $\sigma^2$ denotes the noise variance at each AP. The ergodic rate for UE $k$ is then given by
\begin{equation}
R_k^{\rm erg} = B\mathbb{E}\left[ \log_2\left(1 + \mathrm{SINR}_k^{\mathrm{INS}}\right) \right],
\end{equation}
where $B$ is the system bandwidth. Note that the expectation operator $\mathbb{E}[\cdot]$ lies \emph{outside} the logarithm, which captures the true average mutual information over the channel distribution. 
{The ergodic-rate maximization problem can be formulated as a functional optimization problem over the space of admissible receiving functions:
\begin{equation}   
\begin{aligned}
\max_{\{\mathbf{f}_{km} \in \mathcal{H}_m\}} & \quad \mathcal{J}_k\left(\{\mathbf{f}_{km}\}_{m=1}^M\right) \\
\text{s.t.} & \quad \mathbf{v}_{km} = \mathbf{f}_{km}(\mathbf{H}_m, \mathcal{S}), \quad \forall m \in \mathcal{M},
\end{aligned}
\end{equation}
where $\mathcal{H}_m \triangleq L^2(\mathbb{C}^{N\times K}, p(\mathbf{H}_m); \mathbb{C}^N)$ denotes the Hilbert space of square-integrable functions mapping the local channel matrix to a combining vector, and the objective functional $\mathcal{J}_k: \prod_{m=1}^M \mathcal{H}_m \rightarrow \mathbb{R}$ is defined as
\begin{equation}
    \begin{aligned}
&\mathcal{J}_k\left(\{\mathbf{f}_{km}\}_{m=1}^M\right) \\\triangleq& B\mathbb{E}\left[ \log_2\left(1 + \mathrm{SINR}_k^{\mathrm{INS}}\left(\{\mathbf{f}_{km}(\mathbf{H}_m, \mathcal{S})\}_{m=1}^M\right)\right) \right].
\end{aligned}
\end{equation}
This formulation explicitly reveals that the optimization is performed over an infinite-dimensional function space, rather than a finite-dimensional vector space, which necessitates the use of variational calculus tools for deriving the optimal receiver structure.}
As expected, this problem is non-convex and challenging to solve directly due to the coupling among receive vectors in the SINR expression and the expectation over the channel distribution. 
In the following subsections, we build a systematic solution framework to solve it.
\subsection{Scale Invariance and Large-System Approximation-Based Problem Reformulation}
In this subsection, we introduce a normalization constraint to transform the original fractional objective into an equivalent constrained optimization problem by leveraging the scale invariance of the instantaneous SINR. Recognizing that the resulting pointwise constraint remains analytically intractable, we then invoke large-dimensional random matrix theory to relax it into a statistical (expectation-based) constraint in the large-system regime.
\subsubsection{Scale Invariance and Normalization}
Observe that for any non-zero scalar $c_k$, replacing $\{\mathbf{v}_{km}\}$ with $\{c_k\mathbf{v}_{km}\}$ leaves $\mathrm{SINR}_k^{\mathrm{INS}}$ unchanged. This scale invariance implies that the problem admits infinitely many equivalent solutions differing only by a scaling factor. To eliminate this non-uniqueness, we introduce a normalization constraint by setting the interference-plus-noise term to unity:
\begin{equation} \label{itf_ue_ori}
    \sum_{i=1\backslash k}^K p_i\left|\sum_{m=1}^M\mathbf{v}_{km}^H \mathbf{h}_{im}\right|^2 + \sigma^2 \sum_{m=1}^M\mathbf{v}_{km}^H \mathbf{v}_{km} = 1.
\end{equation}
This constraint is required to hold pointwise for each channel realization. The equivalent formulation of the original problem is thus given by
\begin{subequations}\label{prob_pointwise_ori}
    \begin{align}
        \max_{\{\mathbf{v}_{km}\}_{m=1}^M} \quad &  \mathcal{F}_k(\{\mathbf{v}_{km}\}_{m=1}^M)\label{obj_prob_pointwise_ori} \\
        \text{s.t.} \quad & \eqref{itf_ue_ori},
    \end{align}
\end{subequations}
where
\begin{equation}
    \mathcal{F}_k\left(\{\mathbf{v}_{km}\}_{m=1}^M\right) = \mathbb{E}\left[ \log_2\left(1 + p_k\left| \sum_{m=1}^M \mathbf{v}_{km}^H \mathbf{h}_{km} \right|^2 \right) \right].
\end{equation}
Even under the normalization constraint, the expectation in the objective function remains outside the logarithm and cannot be simplified to a direct maximization of the numerator. Moreover, the pointwise constraint \eqref{itf_ue_ori} is analytically intractable, as the first-order optimality conditions for problem \eqref{prob_pointwise_ori} are generally difficult to solve. 

\subsubsection{Large-Dimensional Asymptotics}
To obtain a tractable formulation, we replace the pointwise normalization with an expectation-based surrogate constraint. In the large-system regime where the total number of antennas $NM$ is large, the channel hardening phenomenon ensures that the random quantities in the SINR denominator concentrate around their expectations. Specifically, for quadratic forms of the type $\sum_{m=1}^M \mathbf{v}_{km}^H\mathbf{h}_{im}$, the fluctuations relative to the mean become negligible as the dimension grows. Consequently, the pointwise constraint can be relaxed to an expectation-based constraint:
\begin{equation}
    \mathbb{E}\left[\sum_{i=1\backslash k}^K p_i\left|\sum_{m=1}^M\mathbf{v}_{km}^H \mathbf{h}_{im}\right|^2 + \sigma^2 \sum_{m=1}^M\mathbf{v}_{km}^H \mathbf{v}_{km}\right] = 1, \label{E_itf_ue_ori}
\end{equation}
which is significantly more tractable for analysis. This approximation becomes asymptotically tight in the large-dimensional limit, with a relative error of order $\mathcal{O}(1/NM)$~\cite{couillet2011CUP}. The original optimization problem can thus be approximated as
\begin{subequations}\label{prob_exp_surrogate}
    \begin{align}
        \max_{\{\mathbf{v}_{km}\}_{m=1}^M} \quad & \mathcal{F}_k(\{\mathbf{v}_{km}\}_{m=1}^M) \\
        \text{s.t.} \quad & \eqref{E_itf_ue_ori}.
    \end{align}
\end{subequations}
The expectation-based constraint \eqref{E_itf_ue_ori} makes the subsequent variational analysis tractable while preserving the essential characteristics of the original problem \eqref{prob_pointwise_ori} in the large-system limit. 


\subsection{Functional-Variational Solution for Asymptotic Problems}\label{subsec:Elog_KKT}
In this subsection, we solve the resulting relaxed problem \eqref{prob_exp_surrogate} using functional-variational analysis, which yields the first-order stationarity conditions and leads to the derivation of the asymptotically optimal Q-LMMSE receiver with a deterministic-equivalent scalar characterization of the cross-AP coupling terms.

\subsubsection{Lagrangian and Perturbation-Based Variational Analysis for Stationarity Conditions}
Since the constraint \eqref{E_itf_ue_ori} is in expectation form, we can construct the Lagrange dual functional of \eqref{prob_exp_surrogate} for UE $k$ as
\begin{equation}
    \begin{aligned}
        &\mathcal{L}_k = \mathcal{F}_k + \mu_k \mathcal{L}_k^{\mathrm{con}},
    \end{aligned}
\end{equation}
where $\mathcal{L}_k\triangleq\mathcal{L}_k(\{\mathbf{v}_{km}\}_{m=1}^M, \mu_k)$, $\mathcal{F}_k\triangleq\mathcal{F}_k(\{\mathbf{v}_{km}\}_{m=1}^M)$, and
$\mathcal{L}_k^{\mathrm{con}} \triangleq \eqref{E_itf_ue_ori}$
denotes the constraint term. Here, $\mu_k$ is the Lagrange multiplier associated with the expectation-based constraint for UE $k$. {To derive the first-order optimality conditions, we then take the functional derivative of $\mathcal{L}_k$ with respect to $\mathbf{v}_{jl}$ for all $j \in \mathcal{K}$ and $l \in \mathcal{M}$, and setting it to zero. }
First, we introduce a perturbation $\epsilon \boldsymbol{\kappa}_{kl}$, where $\epsilon$ is an infinitesimal scalar and $\boldsymbol{\kappa}_{kl} \triangleq \boldsymbol{\kappa}_{kl}(\mathbf{H}_l)$ is an arbitrary differentiable test function satisfying appropriate boundary conditions. The perturbed receive vector at AP $l$ is given by
\begin{equation}
\hat{\mathbf{v}}_{kl} = \mathbf{v}_{kl} + \epsilon \boldsymbol{\kappa}_{kl}, \quad \forall k \in \mathcal{K}, l \in \mathcal{M}.
\end{equation}
The first-order functional derivative of $\mathcal{L}_k$ with respect to $\mathbf{v}_{jl}$ is then defined as
\begin{equation}
\frac{\delta \mathcal{L}_k}{\delta \mathbf{v}_{jl}} = \lim_{\epsilon \to 0} \frac{\mathcal{L}_k(\{{\mathbf{v}}_{km}\}_{m=1}^M, \epsilon \boldsymbol{\kappa}_{jl}, \mu_k) - \mathcal{L}_k(\{{\mathbf{v}}_{km}\}_{m=1}^M, \mu_k)}{\epsilon}. 
\end{equation}
We now derive the first-order variations of the objective term $\mathcal{F}_k$ and the constraint term $\mathcal{L}_k^{\mathrm{con}}$ separately. Since the objective $\mathcal{F}_k$ depends only on $\{\mathbf{v}_{km}\}_{m=1}^M$ for a fixed UE $k$, we have $\frac{\delta \mathcal{F}_k}{\delta \mathbf{v}_{jl}}=0$ for all $j \neq k$. For $j = k$, the first-order variation of $\mathcal{F}_k$ with respect to $\mathbf{v}_{kl}$ is given by
    \begin{align}\label{1st_variation_obj}
        \frac{\delta \mathcal{F}_k}{\delta \mathbf{v}_{kl}}
        =2\Re\left\{ \mathbb{E}\left[ \boldsymbol{\kappa}_{kl}^H \mathbf{h}_{kl} \varphi_k \right]\right\},
    \end{align}
where
\begin{equation}
\varphi_k = \frac{p_k}{\ln 2} \frac{\left( \sum_{m=1}^M \mathbf{v}_{km}^H \mathbf{h}_{km} \right)^H}{1 + p_k \left| \sum_{m=1}^M \mathbf{v}_{km}^H \mathbf{h}_{km} \right|^2},
\end{equation}
Similarly, the first-order variation of the constraint term $\mathcal{L}_k^{\mathrm{con}}$ with respect to $\mathbf{v}_{kl}, j=k$ is
\begin{equation}\label{1st_variation_con}
    \begin{aligned}
        \frac{\delta \mathcal{L}_k^{\mathrm{con}}}{\delta \mathbf{v}_{kl}} = 2\Re\left\{ \mathbb{E}\left[\boldsymbol{\kappa}_{kl}^H \left( \mu_k \sum_{i=1\backslash k}^K p_i \mathbf{h}_{il} \varUpsilon_{ki}^H + \mu_k \sigma^2 \mathbf{v}_{kl} \right) \right] \right\},
    \end{aligned}
\end{equation}
where $\varUpsilon_{ki} = \sum_{m=1}^M \mathbf{v}_{km}^H \mathbf{h}_{im}$. As with the objective term, the constraint variation satisfies $\frac{\delta \mathcal{L}_k^{\mathrm{con}}}{\delta \mathbf{v}_{jl}} = 0$ for all $j \neq k$.
\begin{proof}
    See Appendix~\ref{app:1st_variation}.
\end{proof}

We note that an arbitrary perturbation function actually considerably affects our solution, but fortunately we can eliminate this effect using the following Lemma \ref{lem_vm}.

\begin{lemma}[Fundamental Lemma of Complex-Vector-Valued Variational Calculus]\label{lem_vm}
    Let $(\mathcal{X}, \Sigma, \mu)$ be a measure space with $\mu(\mathcal{X}) < \infty$. Consider the Hilbert space of square-integrable functions $\mathcal{H} := L^2(\mathcal{X}, \mu; \mathbb{C}^N)$ with inner product $\langle \mathbf{u}, \mathbf{v} \rangle := \int_{\mathcal{X}} \mathbf{u}(\mathbf{x})^H \mathbf{v}(\mathbf{x}) \, d\mu(\mathbf{x})$. If $\mathbf{f} \in \mathcal{H}$ satisfies $\Re\{ \langle \boldsymbol{\eta}, \mathbf{f} \rangle \} = 0$ for every test function $\boldsymbol{\eta} \in \mathcal{H}$, then $\mathbf{f}(\mathbf{x}) = \mathbf{0}$ for $\mu$-almost every $\mathbf{x} \in \mathcal{X}$.
\end{lemma}
\begin{proof}
    See Appendix~\ref{app:lemma}.
\end{proof}

According to {Lemma \ref{lem_vm}}, since $\boldsymbol{\kappa}_{kl}(\mathbf{H}_l)$ is an arbitrary perturbation direction, the stationary point must satisfy $\frac{\delta \mathcal{L}_k}{\delta \mathbf{v}_{kl}} =0, \forall l$, which leads to the following set of equations:
\begin{equation}\label{first-order-deltaLk}
    \begin{aligned}
       &\mathbb{E}\Bigg[
        \mathbf{h}_{kl}\varphi_k +  \mu_k \left( \sum_{i=1\backslash k}^K p_i \mathbf{h}_{il} \varUpsilon_{ki}^H + \sigma^2 \mathbf{v}_{kl} \right)
        \,\Bigg|\, \mathbf{H}_l
    \Bigg] = \mathbf{0},\forall l, 
    \end{aligned}
\end{equation}
where $\mathbb{E}[\cdot \mid \mathbf{H}_l]$ denotes the conditional expectation given $\mathbf{H}_l$.
\begin{proof}
    See Appendix~\ref{app:station_point}.
\end{proof}

\subsubsection{Asymptotically Optimal Local Receiver Structure}
We hereafter decompose the aggregate term $\sum_{m=1}^M \mathbf{v}_{km}^H \mathbf{h}_{im}$ into a local component at AP $l$ and a cross-AP coupling component:
\begin{equation}
\sum_{m=1}^M \mathbf{v}_{km}^H \mathbf{h}_{im} = \mathbf{v}_{kl}^H \mathbf{h}_{il} + \sum_{m \neq l}^M \mathbf{v}_{km}^H \mathbf{h}_{im}.
\end{equation}
Substituting this decomposition into $\varUpsilon_{ki}$ in \eqref{first-order-deltaLk} and rearranging terms, we obtain
\begin{equation}\label{eq:decomposed_stationary}
    \begin{aligned}
    &\mu_k \left( \sum_{i=1\backslash k}^K p_i \mathbf{h}_{il}\mathbf{h}_{il}^H + \sigma^2 \mathbf{I}_N \right) \mathbf{v}_{kl}\\
        = &- \mathbb{E}\Bigg[
        \mathbf{h}_{kl}\varphi_k + \mu_k \sum_{i=1\backslash k}^K p_i \mathbf{h}_{il} \left( \sum_{m \neq l}^M \mathbf{v}_{km}^H \mathbf{h}_{im} \right)^H
        \,\Bigg|\, \mathbf{H}_l
    \Bigg],
    \end{aligned}
\end{equation}
where the left-hand side follows from the fact that $\mathbf{v}_{kl}$ is a function of $\mathbf{H}_l$ and thus can be extracted from the conditional expectation. Solving for $\mathbf{v}_{kl}$, we obtain the Q-LMMSE receive vector structure
\begin{equation}\label{q-lmmse str}
    \begin{aligned}
   & \mathbf{v}_{kl}
    = - \left( \sum_{i=1\backslash k}^K p_i \mathbf{h}_{il}\mathbf{h}_{il}^H + \sigma^2 \mathbf{I}_N \right)^{-1} \times\\
    &\left(\mathbf{h}_{kl}\eta_{kl}(\mathbf{H}_l)+\mathbb{E}\Bigg[
        \sum_{i=1\backslash k}^K p_i \mathbf{h}_{il} \left( \sum_{m \neq l}^M \mathbf{v}_{km}^H \mathbf{h}_{im} \right)^H
        \,\Bigg|\, \mathbf{H}_l
    \Bigg]\right),
    \end{aligned}
\end{equation}
where
\begin{equation}
\eta_{kl}(\mathbf{H}_l) \triangleq \mathbb{E}\left[
        \frac{p_k}{\mu_k \ln 2} \frac{\left( \sum_{m=1}^M \mathbf{v}_{km}^H \mathbf{h}_{km} \right)^H}{1 + p_k \left| \sum_{m=1}^M \mathbf{v}_{km}^H \mathbf{h}_{km} \right|^2}
\,\Bigg|\, \mathbf{H}_l \right].
\end{equation}
\eqref{q-lmmse str} reveals that the receive vector consists of three components: an MMSE-type inverse matrix, a desired-signal term scaled by $\eta_{kl}(\mathbf{H}_l)$, and a cross-AP interference coupling term. 

\subsubsection{Large-System Random Matrix Analysis for Cross-AP Coupling Terms}
The term $\left( \sum_{m \neq l}^M \mathbf{v}_{km}^H \mathbf{h}_{im} \right)^H$ represents the cross-AP coupling effect, which captures the influence of receive vectors at other APs on the optimal local receive vector at AP $l$. This coupling term is generally intractable in closed form due to its dependence on the receive vectors and channel realizations across all APs. However, in the large-system regime where $N$ and $M$ grow large with a fixed ratio, random matrix theory provides powerful tools to approximate such random functionals by deterministic equivalents that depend only on channel statistics. This enables an explicit characterization of the optimal receiver structure.

\begin{lemma}\label{lemma_appro}
   For UE $k$ and AP $l$, and for any $i \neq k$ and $m \neq l$, the cross-AP coupling terms satisfy
\begin{equation}
    \mathbb{E}\left[
        \left( \sum_{m \neq l}^M \mathbf{v}_{km}^H \mathbf{h}_{im} \right)^H
        \,\middle|\, \mathbf{H}_l
    \right] = 0,
\end{equation}
and consequently,
\begin{equation}
\mathbb{E}\left[
        \sum_{i=1\backslash k}^K p_i \mathbf{h}_{il}
        \left( \sum_{m \neq l}^M \mathbf{v}_{km}^H \mathbf{h}_{im} \right)^H
        \,\middle|\, \mathbf{H}_l
    \right] = \mathbf{0}.
\end{equation}
\end{lemma}
\begin{proof}
    See Appendix~\ref{app:lemma_appro}.
\end{proof}
By Lemma~\ref{lemma_appro}, the cross-AP coupling terms vanish in the conditional mean sense. Substituting this result into \eqref{q-lmmse str}, the local receive vector simplifies to
\begin{equation}\label{mmse_like_2}
    \begin{aligned}
     \mathbf{v}_{kl}
    = &- \left( \sum_{i=1\backslash k}^K p_i \mathbf{h}_{il}\mathbf{h}_{il}^H + \sigma^2 \mathbf{I}_N \right)^{-1} \mathbf{h}_{kl} \eta_{kl}(\mathbf{H}_l).
    \end{aligned} 
\end{equation}
This result indicates that although the cross-AP coupling terms are generally non-zero for individual channel realizations, their conditional expectation vanishes, thereby decoupling the local receive vector design across APs. 

The scalar term $\eta_{kl}(\mathbf{H}_l)$ involves high-dimensional quadratic forms and fractional random functionals. Under standard random matrix concentration results, its dependence on the instantaneous channel realization vanishes asymptotically. Specifically, when the total number of antennas $NM$ is sufficiently large, there exists a deterministic scalar $\bar{\eta}_k$ (depending only on the UE index $k$) such that
\begin{equation}
    \eta_{kl}(\mathbf{H}_l) - \bar{\eta}_k \xrightarrow{NM \to \infty} 0,
    \qquad \forall l \in \mathcal{M}.
\end{equation}
Consequently, \eqref{mmse_like_2} can be further approximated as
\begin{equation}
    \mathbf{v}_{kl}
    \approx
    - \bar{\eta}_k
    \left( \sum_{i=1\backslash k}^K p_i \mathbf{h}_{il}\mathbf{h}_{il}^H + \sigma^2 \mathbf{I}_N \right)^{-1}\mathbf{h}_{kl}.
\end{equation}
This reveals that, for a fixed UE $k$, the local receive vectors at all APs converge to the same MMSE-type structure, differing only by a deterministic scaling factor $\bar{\eta}_k$. Since this common scaling factor does not affect the instantaneous SINR due to the scale invariance, it can be absorbed into the normalization. Thus, in the large-system limit, the asymptotically optimal receive vector Q-LMMSE is given by
\begin{equation}\label{qlmmse}
    \mathbf{v}_{kl}^{\mathrm{Q-LMMSE}}
    = \left( \sum_{i=1\backslash k}^K p_i \mathbf{h}_{il}\mathbf{h}_{il}^H + \sigma^2 \mathbf{I}_N \right)^{-1}\mathbf{h}_{kl}.
\end{equation}

\begin{remark}{[Physical Interpretation and Relationship with LMMSE Receiver]\label{rem:qlmmse_lmmse}
The Q-LMMSE receiver \eqref{qlmmse} differs from the conventional LMMSE receiver \eqref{v_lmmse_local} only in the interference covariance matrix: Q-LMMSE excludes the desired user's own channel ($i \neq k$), whereas LMMSE includes all users ($i = 1, \ldots, K$). By the Sherman--Morrison formula, the two receivers are related as
\begin{equation}
    \mathbf{v}_{kl}^{\mathrm{Q-LMMSE}} = \frac{\mathbf{v}_{kl}^{\mathrm{LMMSE}}}{1 - p_k \mathbf{h}_{kl}^H \mathbf{W}_l^{-1} \mathbf{h}_{kl}},
\end{equation}
where $\mathbf{W}_l = \sum_{i=1}^K p_i \mathbf{h}_{il}\mathbf{h}_{il}^H + \sigma^2 \mathbf{I}_N$. This reveals that the two receive vectors share the \emph{same direction} but differ by an instantaneous channel-dependent scalar coefficient. }

{
The performance gap between the two architectures stems not from the local combining direction, but from the CPU-side fusion strategy. The Q-LMMSE receiver employs equal-gain combining (direct summation) at the CPU, where the instantaneous scalar coefficient $\frac{1}{1 - p_k \mathbf{h}_{kl}^H \mathbf{W}_l^{-1} \mathbf{h}_{kl}}$ varies across APs and implicitly provides adaptive weighting based on instantaneous channel conditions. In contrast, the LMMSE-LSFD architecture relies on statistical LSFD coefficients computed from long-term channel statistics, which cannot track instantaneous channel variations. Consequently, when evaluated under the ergodic rate $\mathbb{E}[\log_2(1+\mathrm{SINR}^{\mathrm{INS}})]$, the Q-LMMSE receiver achieves superior performance by better exploiting instantaneous spatial degrees of freedom, while maintaining the same per-AP computational complexity and eliminating the CPU-side LSFD computation overhead.}

{It is worth noting that \eqref{qlmmse} is proportional to the classical local instantaneous SINR-maximizing combiner. However, its emergence here is fundamentally different: it is derived as the asymptotically optimal solution to the \emph{ergodic-rate maximization} problem under distributed local CSI constraints, rather than being postulated from instantaneous SINR or MSE criteria. The functional-variational derivation, the large-system decoupling proof, and the elimination of CPU-side LSFD processing constitute the core theoretical contributions of this work.}
\end{remark}

\section{Numerical Results}\label{sec:simu}
In this section, we present numerical results to validate the theoretical analysis and evaluate the performance of the proposed Q-LMMSE receiver. We compare it against the conventional LMMSE-LSFD benchmark under various system configurations. 

\subsection{System Configuration}
We consider a $1 \times 1$~km$^2$ CF-mMIMO system. The APs and UEs are uniformly and independently distributed within this area, unless otherwise specified. The noise power is $\sigma^2 = -94$~dBm (corresponding to a noise power spectral density of $-174$~dBm/Hz and a bandwidth of $100$~MHz).
The key simulation parameters are summarized in Table~\ref{tab:sim_parameters}. All results are averaged over $10^3$ independent channel realizations unless otherwise noted. The LSFD statistics for LMMSE-LSFD are estimated via CPU-side Monte Carlo simulation with $T = 1000$ synthetic realizations.
\begin{table}[t]
    \centering
    \setlength{\tabcolsep}{14pt}
    \caption{Simulation Parameters}
    \label{tab:sim_parameters}
    \begin{tabular}{lc}
        \hline
        Parameter & Value \\
        \hline
        Coverage area, $r$ & $1$ km \\
        Number of users, $K$& $8\sim64$\\
    Number of APs, $M$ & $10\sim{40}$\\
    Number of antennas per AP, $N$ & $8,16,24,32$\\
    Pilot power, $\eta_k$ & $0.4$ W\\
    Uplink transmit power per UE, $p_k$ & $0.01\sim6.25$ W\\
        Shadowing standard deviation, $\sigma_{\mathrm{sh}}$ & $8$ dB \\
        Noise power, $\sigma^2$ & $ -94$ dBm \\
        System bandwidth, $B$ & ${100}$ MHz \\
        Number of channel realizations & $1000$ \\
        \hline
    \end{tabular}
\end{table}

\subsubsection{Channel Model}
The channel vector $\mathbf{h}_{km} \in \mathbb{C}^{N}$ from user $k$ to AP $m$ is modeled as
\begin{equation}
    \mathbf{h}_{km} = \sqrt{\beta_{km}} \, \widetilde{\mathbf{h}}_{km},
\end{equation}
where $\beta_{km}$ represents the large-scale fading coefficient (including pathloss and shadowing), and $\widetilde{\mathbf{h}}_{km} \sim \mathcal{CN}(\mathbf{0}, \mathbf{I}_{N})$ denotes the small-scale fading component. The large-scale fading coefficient is given by
\begin{equation}
    \beta_{km} = \mathrm{PL}_{km} \cdot 10^{\sigma_{\mathrm{sh}} z_{km}/10},
\end{equation}
where $\mathrm{PL}_{km}$ is the pathloss, $\sigma_{\mathrm{sh}}$ is the shadowing standard deviation, and $z_{km} \sim \mathcal{N}(0,1)$ is the shadowing fading term. The pathloss model follows \cite{ETSI2020TR12} 
\begin{equation}
    \mathrm{PL}_{km} = -35.4+34\log_{10}(d_0)+20\log_{10}(f_c),
\end{equation}
where $d_0=50$ m is the reference distance and $f_c=3$ GHz is the carrier frequency. 

\subsubsection{Performance Metric of Average Ergodic Rate}
    The true achievable rate, averaged over all users, computed via Monte Carlo averaging over $10^3$ independent channel realizations:
    \begin{equation}
        \overline{R}_{\mathrm{erg}} = \frac{1}{K} \sum_{k=1}^{K} B\mathbb{E}\!\left[\log_2\!\left(1+\mathrm{SINR}_k^{\mathrm{INS}}\right)\right],
    \end{equation}
where the $\mathrm{SINR}_k^{\mathrm{INS}}$  computed with the Q-LMMSE receive vector is denoted as \textbf{Erg-Q-LMMSE}, and this with the LMMSE-LSFD receive combiner as \textbf{Erg-LMMSE-LSFD}. Also, the global MMSE receiver that optimizes the ergodic rate is denoted as \textbf{Erg-C-MMSE}.
This metric serves as the ground-truth performance indicator, reflecting the true mutual information achievable by the system under the given receiver design.    

\subsection{Performance Evaluation}
\label{subsec:performance_evaluation}

The proposed Q-LMMSE receiver is evaluated against the LMMSE-LSFD benchmark under both ergodic and UatF metrics, across various independent parameter settings. Unless otherwise specified, the default configuration is $M=20$ APs, $N=16$ antennas per AP, $K=20$ users, and uplink transmit power $p_k=1$ W.

\subsubsection{Impact of the Number of APs}
\begin{figure}[!t]
    \centering
    \includegraphics[width=\linewidth]{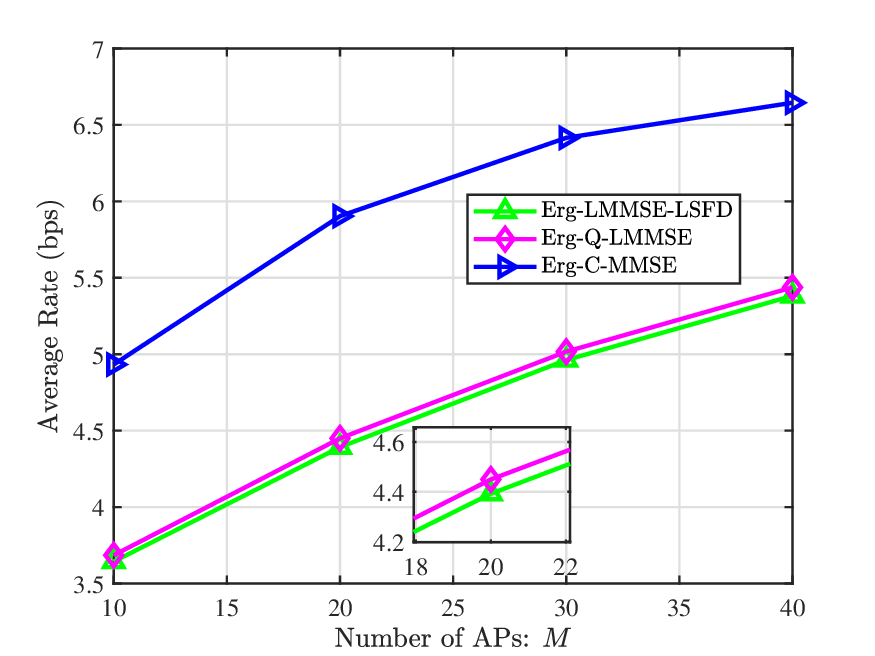}
    \caption{Average ergodic and UatF rates versus the number of APs $M$ ($N=16$, $K=16$, $p_k=1$ W).}
    \label{fig:AP_sweep}
\end{figure}

Fig.~\ref{fig:AP_sweep} illustrates the average rates versus the number of APs $M \in \{10, 20, 30, 40\}$. Several key observations can be made. First, all four curves exhibit a monotonically increasing trend with $M$, which is expected since deploying more APs provides additional macro-diversity and coherent combining gain. Second, under the ergodic rate metric, the proposed \textbf{Erg-Q-LMMSE} approach consistently outperforms the \textbf{Erg-LMMSE-LSFD} benchmark, with the absolute gap widening as $M$ increases. This indicates that the Q-LMMSE receiver more effectively exploits the additional spatial degrees of freedom provided by a denser AP deployment. Third, the \textbf{Erg-C-MMSE} curve, which represents the optimal ergodic rate achievable by a global MMSE receiver, serves as an upper bound. 


\subsubsection{Impact of the Number of Antennas per AP}
\begin{figure}[!t]
    \centering
    \includegraphics[width=\linewidth]{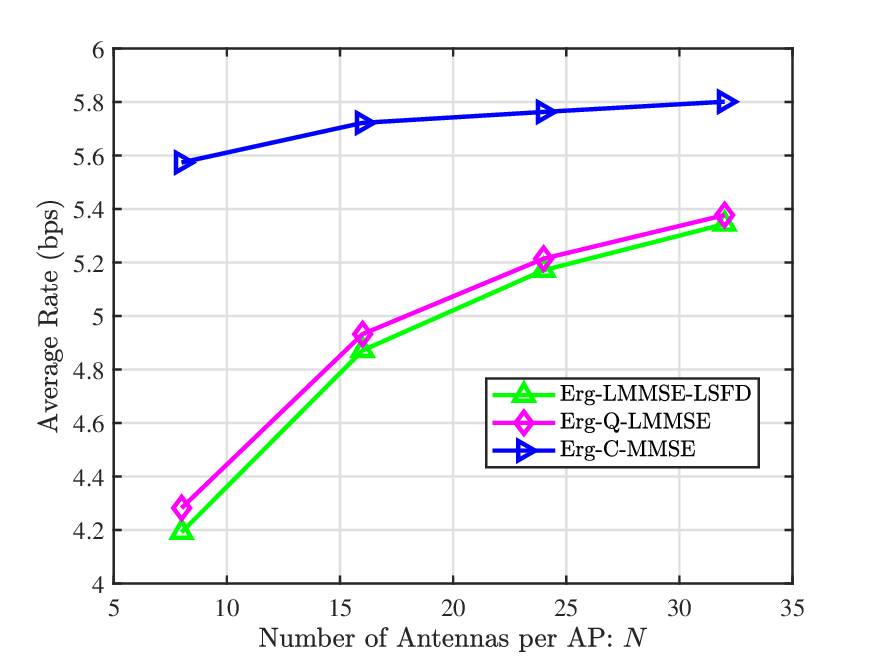}
    \caption{Average ergodic and UatF rates versus the number of antennas per AP $N$ ($M=20$, $K=16$, $p_k=1$ W).}
    \label{fig:An_sweep}
\end{figure}

Fig.~\ref{fig:An_sweep} presents the average rate performance as the number of antennas per AP $N$ varies from $8$ to $32$. The ergodic rates increase significantly with $N$ for all schemes, reflecting the array gain from additional antennas. Notably, the average rate gap between \textbf{Erg-Q-LMMSE} and \textbf{Erg-LMMSE-LSFD} is most pronounced at smaller $N$ (e.g., $N=8$) and gradually narrows as $N$ grows. This convergence behavior is consistent with the channel hardening phenomenon: as $N \to \infty$, the random channel vectors become asymptotically orthogonal, and the instantaneous SINR concentrates around its deterministic equivalent, causing the UatF bound to become increasingly tight. In the practical regime of $N \leq 32$, however, the proposed receiver still delivers a meaningful ergodic-rate gain. The \textbf{Erg-C-MMSE} curve again serves as an upper bound, and the gap between \textbf{Erg-Q-LMMSE} and \textbf{Erg-C-MMSE} reduces with $N$, indicating that the proposed receiver tends to approach optimal performance as the number of antennas increases.

\subsubsection{Impact of the Number of Users}
\begin{figure}[!t]
    \centering
    \includegraphics[width=\linewidth]{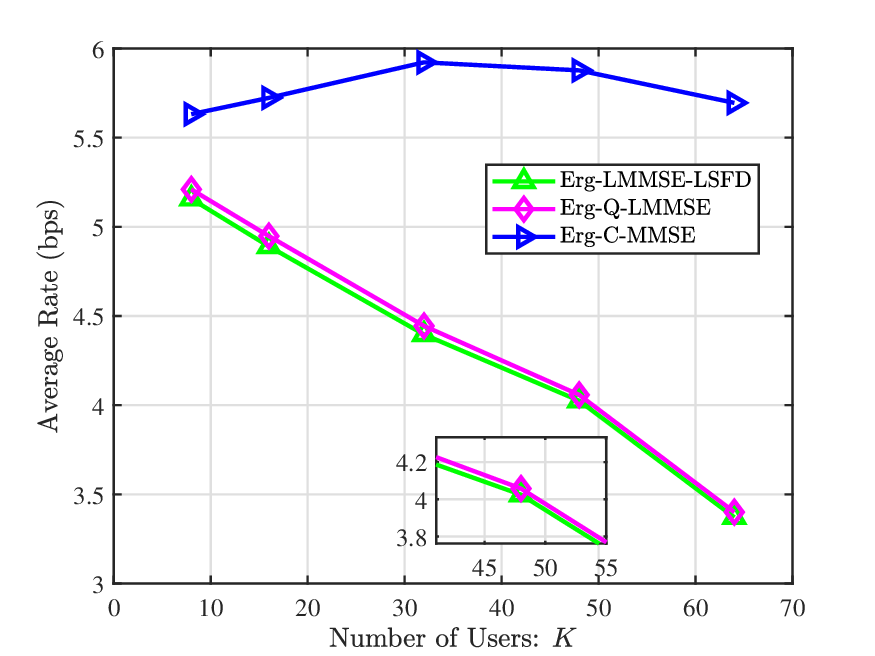}
    \caption{Average ergodic and UatF rates versus the number of users $K$ ($M=20$, $N=16$, $p_k=1$ W).}
    \label{fig:UE_sweep}
\end{figure}

Fig.~\ref{fig:UE_sweep} shows the average rates as the number of users $K$ increases from $8$ to $64$. As observed, the average rate decreases with $K$ for both the \textbf{Erg-Q-LMMSE} and \textbf{Erg-LMMSE-LSFD} schemes.
{However, the \textbf{Erg-C-MMSE} curve exhibits a non-monotonic behavior, initially increasing with $K$ before declining. 
This phenomenon can be attributed to the global CSI availability at the CPU, which enables joint interference suppression across all APs. 
In the low-user regime, increasing $K$ allows the centralized receiver to better exploit the spatial degrees of freedom and multiuser diversity, leading to an initial performance gain. 
As $K$ continues to grow, the system becomes interference-limited, and the average rate eventually decreases. 
In contrast, the distributed schemes lack global CSI coordination, and thus the accumulated interference at each AP increases monotonically with $K$, resulting in a consistent rate degradation.} Despite this overall decline, the relative ergodic-rate advantage of \textbf{Erg-Q-LMMSE} over \textbf{Erg-LMMSE-LSFD} remains positive across the entire range of $K$. Interestingly, the absolute gap appears to be largest at moderate user loads (e.g., $K=16 \sim 32$) and reduces slightly at very high loads ($K=64$), where the system becomes heavily interference-limited and both receivers approach the same performance floor. 

\subsubsection{Impact of the Uplink Transmit Power}
\begin{figure}[!t]
    \centering
    \includegraphics[width=\linewidth]{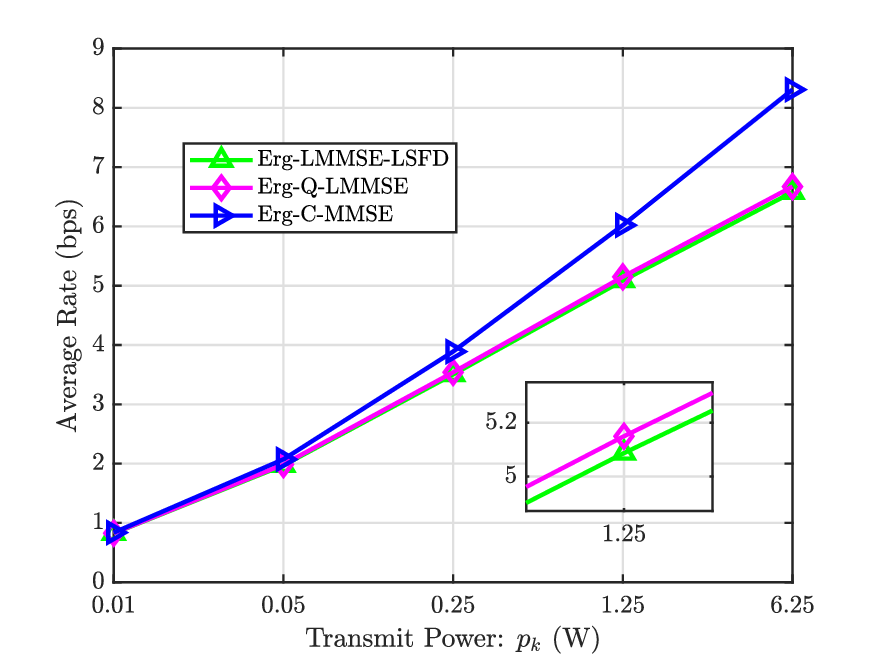}
    \caption{Average ergodic and UatF rates versus the uplink transmit power $p_k$ ($M=20$, $N=16$, $K=16$).}
    \label{fig:Rho_sweep}
\end{figure}

Fig.~\ref{fig:Rho_sweep} depicts the average rate performance as the uplink transmit power $p_k$ varies over a wide range with a $5$-fold increment step. At low power levels, the system is noise-limited, and all rates increase approximately linearly with $p_k$ (on a logarithmic scale). As $p_k$ increases, the system transitions into an interference-limited regime, and the rates begin to saturate. Notably, the curves for \textbf{Erg-Q-LMMSE} and \textbf{Erg-LMMSE-LSFD} remain very close to each other across the entire power range, indicating that the proposed distributed receiver achieves performance comparable to the UatF-optimal benchmark. However, the gap between the distributed schemes and the centralized \textbf{Erg-C-MMSE} scheme widens as the transmit power increases. This is because the centralized receiver can exploit global CSI to perform joint interference suppression, which becomes increasingly critical in the high-SNR interference-limited regime. In contrast, the distributed receivers are constrained by local CSI, limiting their ability to fully eliminate inter-user interference at high power levels.

\section{{Conclusion}}
\label{sec:conclusion}


This paper investigated the design of uplink local receivers in CF-mMIMO systems with localized processing, addressing the fundamental gap between optimizing the true ergodic achievable rate and its UatF lower bound. Under the assumption of perfect local CSI, we developed a functional-variational framework for direct ergodic-rate-oriented receiver design and derived an asymptotically optimal Q-LMMSE receiver in the large-system regime, where cross-AP coupling terms vanish in the conditional mean. 
Via the Sherman--Morrison formula, we showed that the proposed Q-LMMSE receiver maintains the same per-AP complexity as standard LMMSE receiver while completely eliminating the CPU-side LSFD coefficients calculation. Numerical results across various system configurations validated that the Q-LMMSE receiver consistently outperforms the LMMSE-LSFD benchmark in terms of the ergodic rate, demonstrating a compelling trade-off between superior ergodic performance and strictly lower system-level complexity. Future work will extend this analysis to imperfect CSI scenarios, investigate the impact of fronthaul capacity constraints, and explore robust receiver designs under practical hardware impairments.

\begin{appendices}

\section{The optimality of the UatF-Based Receiver}\label{app:uatf_proof}

\subsection{Problem Formulation for UatF-Based Rate Maximization}

Since the UatF-based rate $R_k^{\text{UatF}}=B{\rm log}_2(1+\mathrm{SINR}_k^{\text{UatF}})$ is a monotonically increasing function of $\mathrm{SINR}_k^{\text{UatF}}$, maximizing the rate is equivalent to maximizing the UatF SINR directly. Furthermore, by exploiting the scale invariance of the SINR and normalizing the interference-plus-noise power to unity, the optimization problem for UE $k$ can be equivalently formulated as
\begin{subequations}\label{sinr_max_vkm}
    \begin{align}
    \max_{\{{\mathbf{v}_{km}}\}_{m=1}^M} \quad & {p_k}\bigg| \sum_{m=1}^M \mathbb{E}[\mathbf{v}_{km}^H \mathbf{h}_{km}] \bigg|^2 \label{obj_sinr} \\
    \text{s.t.} \quad & {\rm INT}_k = 1, \label{const_den}
    \end{align}
\end{subequations} 
where ${\rm INT}_k \triangleq \sum_{i=1}^K {p_i}\mathbb{E}\left[ \big| \sum_{m=1}^M\mathbf{v}_{km}^H \mathbf{h}_{im} \big|^2 \right] - {p_k}\left| \sum_{m=1}^M \mathbb{E}\big[\mathbf{v}_{km}^H {\mathbf{h}}_{km}\big] \right|^2 + \sigma^2 \sum_{m=1}^M \mathbb{E}\left[\mathbf{v}_{km}^H \mathbf{v}_{km} \right]$ denotes the interference-plus-noise power, while the objective function \eqref{obj_sinr} represents the desired signal power for UE $k$. Note that the LSFD weights $a_{km}$ are omitted here since they appear as a common factor in both the numerator and denominator of \eqref{sinr_uatf}, and thus do not affect the optimization \eqref{sinr_max_vkm}. 

\subsection{Functional-variational Analysis for Optimal Local MMSE Receiver}

The optimal local receive vector under the UatF framework is the local MMSE receiver, which can be obtained by optimizing \eqref{sinr_max_vkm}. 
We observe that the objective function in \eqref{sinr_max_vkm} is a quadratic functional of the receive vectors, rather than a logarithmic functional.
This leads to a more straightforward derivation of the optimal receiver structure, which coincides with the classical local MMSE receiver.

\subsubsection{Lagrangian Functional and Variational Solution}
Following the approach in Section~\ref{subsec:Elog_KKT}, we construct the Lagrangian functional for problem \eqref{sinr_max_vkm} as
\begin{equation}\label{Lagr}
    \mathcal L_k\left(\{\mathbf{v}_{km}\}_{m=1}^M,\,{\mu_{k}}\right) = 
    {p_k}\bigg| \sum_{m=1}^M \mathbb{E}[\mathbf{v}_{km}^H {\mathbf{h}}_{km}] \bigg|^2 
    + {\mu_{k}} ({\rm INT}_k - 1), 
\end{equation}
where ${\mu_{k}}$ is the Lagrange multiplier for UE $k$. The corresponding functional derivative with respect to $\mathbf{v}_{kl}$ is given by
\begin{equation}\label{var_der}
    \begin{aligned}
    &\frac{\delta \mathcal L_k}{\delta \mathbf{v}_{kl}} = 
    2\Re\left\{ p_k\mathbb{E}[\mathbf{q}_{kl}^H {\mathbf{h}}_{kl}] \left(\sum_{m=1}^M \mathbb{E}[\mathbf{v}_{km}^H {\mathbf{h}}_{km}]\right)^{H} \right\}\\
    +&2\mu_k\Re\left\{p_i\sum_{i=1}^K \mathbb{E}\left[ \mathbf{q}_{kl}^{H}{\mathbf{h}}_{il} \left(\sum_{m=1}^M {\mathbf{v}_{km}^H\mathbf{h}}_{im}\right)^H \right]\right\}
    \\
    -&2\mu_k\Re\left\{ p_k\mathbb{E}[\mathbf{q}_{kl}^H {\mathbf{h}}_{kl}] \left(\sum_{m = 1}^M \mathbb{E}[\mathbf{v}_{km}^H {\mathbf{h}}_{km}]\right)^{H} \right\}\\
    &+2\mu_k\sigma^2\Re\left\{ \mathbb{E}\left[ \mathbf{q}_{kl}^{H}\mathbf{v}_{kl}  \right]\right\}.
    \end{aligned} 
\end{equation}
where $\mathbf{q}_{kl}$ is an arbitrary differentiable test function satisfying appropriate boundary conditions. Setting the functional derivative to zero and applying Lemma~\ref{lem_vm}, we obtain
\begin{equation}\label{variMeth}
   \begin{aligned}
        & p_k{\mathbf{h}}_{kl} \sum_{m = 1}^M \left(\mathbb{E}[\mathbf{v}_{km}^H {\mathbf{h}}_{km}]\right)^{H}
        + \mu_{k} 
            \sum_{i=1}^K p_i{\mathbf{h}}_{il} {\mathbf{h}}_{il}^H\mathbf{v}_{kl}\\
            &+ \mu_{k} \sum_{i=1}^K p_i \mathbf{h}_{il} \sum_{m \neq l}^M \left(\mathbb{E}[{\mathbf{h}}_{im}^H\mathbf{v}_{km}]\right)^{H}\\
            &- \mu_{k}p_k{\mathbf{h}}_{kl} \sum_{m = 1}^M \left(\mathbb{E}[\mathbf{v}_{km}^H {\mathbf{h}}_{km}]\right)^{H} 
            +\mu_{k}\sigma^2 \mathbf{v}_{kl} =0.
   \end{aligned} 
\end{equation} 
To this end, the uplink SINR$_{k}^{\rm UatF}$ in \eqref{sinr_uatf} with fixed LSFD weights is maximized by an optimal uplink local receiver via functional optimization:
\begin{equation}\label{v_full}
    \begin{aligned}
    \mathbf{v}_{kl} =& \left(\sum_{i=1}^K p_i{\mathbf{h}}_{il} {\mathbf{h}}_{il}^H + \sigma^2 \mathbf{I}\right)^{-1}{\mathbf{h}}_{kl} \\
    &\times p_k\left(\mathbb{E}[{\mathbf{h}}_{kl}^H\mathbf{v}_{kl} ]-\frac{1}{\mu_k}\sum_{m = 1}^M \mathbb{E}[{\mathbf{h}}_{km}^H\mathbf{v}_{km} ]\right) .
    \end{aligned}
\end{equation}
Since this component $p_k(\mathbb{E}[{\mathbf{h}}_{kl}^H\mathbf{v}_{kl} ]-\frac{1}{\mu_k}\sum_{m = 1}^M \mathbb{E}[{\mathbf{h}}_{km}^H\mathbf{v}_{km} ])$ is a scalar and independent of the instantaneous CSI, it can be precomputed at the CPU. In fact, this coefficient corresponds exactly to the LSFD coefficient. However, as its direct calculation here is impractical, we leverage existing LSFD methods for its computation. Therefore, the final local receiver is
\begin{equation}\label{receComb}
    \mathbf v_{kl}={\bigg(\sigma^2\mathbf{I}_N+\sum_{i=1}^Kp_i{\mathbf h}_{il}{\mathbf h}^H_{il}\bigg)^{-1}{\mathbf h}_{kl}}, \forall k,l.
\end{equation}
Based on the optimal local receive vector, the optimal LSFD coefficients is denoted as \eqref{a_lsfd_opt} \cite{BJORNSON2021Book}.

\section{Proof of \eqref{1st_variation_obj} and \eqref{1st_variation_con}}\label{app:1st_variation}
{\bf{Proof:}} We will derive the first-order variation of the objective term $\mathcal{F}_k$ and the constraint term $\mathcal{L}_k^{\mathrm{con}}$ w.r.t. $\mathbf{v}_{jl}, \forall j\in\{1,\dots,K\},l\in\{1,\dots,M\}$ separately.
\subsection{First-order variation of the objective term $\mathcal{F}_k$}
For the case $j=k$, the first-order variation of $\mathcal{F}_k$ w.r.t. $\mathbf{v}_{kl}$ is given by
\begin{equation}\label{delta_F_app}
    \begin{aligned}
        \frac{\delta \mathcal{F}_k}{\delta \mathbf{v}_{kl}}
        =& \lim_{\epsilon \to 0} \frac{\mathcal{F}_k(\{{\mathbf{v}}_{km}\}_{m=1}^M, \epsilon \boldsymbol{\kappa}_{kl}) - \mathcal{F}_k(\{{\mathbf{v}}_{km}\}_{m=1}^M)}{\epsilon} \\
        =& \mathbb{E}\left[\lim_{\epsilon \to 0} \frac{1}{\epsilon}\left(f\left(\Gamma_k+\epsilon\mathbf{h}_{kl}^H \boldsymbol{\kappa}_{kl}\right) - f\left(\Gamma_k\right) \right)\right],
    \end{aligned}
\end{equation}
where
     $f\left(\varGamma_k\right) = \log_2\left(1 + p_k\left| \varGamma_k \right|^2 \right),$
with $\varGamma_k = \sum_{m=1}^M \mathbf{v}_{km}^H \mathbf{h}_{km}$. By applying the first-order Taylor expansion, we can further derive
\begin{align}\label{f_gamma_epsilon}
    \begin{aligned}
    &f\left(\varGamma_k+\epsilon\mathbf{h}_{kl}^H \boldsymbol{\kappa}_{kl}\right)\\
    &\approx f\left(\varGamma_k\right) + \frac{\partial f}{\partial \varGamma_k}\epsilon \mathbf{h}_{kl}^H \boldsymbol{\kappa}_{kl}+ \frac{\partial f}{\partial \varGamma_k^*}\epsilon \boldsymbol{\kappa}_{kl}^H\mathbf{h}_{kl} + o(\epsilon),
    \end{aligned}  
\end{align}
where $\frac{\partial f}{\partial \varGamma_k} = \frac{p_k}{\ln 2}\frac{\varGamma_k^*}{1 + p_k|\varGamma_k|^2}$ and $\frac{\partial f}{\partial \varGamma_k^*} = \left(\frac{\partial f}{\partial \varGamma_k}\right)^*$. By substituting \eqref{f_gamma_epsilon} into \eqref{delta_F_app}, we have
\begin{subequations}
  \begin{align}
\frac{\delta \mathcal{F}_k}{\delta \mathbf{v}_{kl}} \approx& \mathbb{E}\left[\lim_{\epsilon \to 0} \frac{1}{\epsilon}\left( \frac{\partial f}{\partial \varGamma_k}\epsilon \mathbf{h}_{kl}^H \boldsymbol{\kappa}_{kl}+ \frac{\partial f}{\partial \varGamma_k^*}\epsilon \boldsymbol{\kappa}_{kl}^H\mathbf{h}_{kl} + o(\epsilon) \right)\right] \\
=&\frac{2p_k}{\ln 2}\Re\left\{ \mathbb{E}\left[ \frac{\boldsymbol{\kappa}_{kl}^H \mathbf{h}_{kl}\varGamma_k}{(1 + p_k|\varGamma_k|^2)} \right]\right\},\label{re_res_obj}
\end{align}  
\end{subequations}
where \eqref{re_res_obj} is obtained due to the fact that $\varGamma_k^*\mathbf{h}_{kl}^H \boldsymbol{\kappa}_{kl} = (\boldsymbol{\kappa}_{kl}^H \mathbf{h}_{kl}\varGamma_k)^*$, and hence the two terms in the numerator are complex conjugates of each other.
Then, by substituting $\varGamma_k = \sum_{m=1}^M \mathbf{v}_{km}^H \mathbf{h}_{km}$, we obtain the desired result in \eqref{1st_variation_obj}.

\subsection{First-order variation of the constraint term $\mathcal{L}_k^{\mathrm{con}}$}
For the case $j=k$, the first-order variation of $\mathcal{L}_k^{\mathrm{con}}$ w.r.t. $\mathbf{v}_{kl}$ is given by
\begin{equation}\label{delta_con_app}
    \begin{aligned}
        &\frac{\delta \mathcal{L}_k^{\mathrm{con}}}{\delta \mathbf{v}_{kl}}
        =\mu_k\mathbb{E}\left[\lim_{\epsilon \to 0} \frac{1}{\epsilon}\left(g\left(\varUpsilon_{ki}+\epsilon \boldsymbol{\kappa}_{kl}^H\mathbf{h}_{il}\right) - g\left(\varUpsilon_{ki}\right) \right)\right] \\
        +&\mu_k\sigma^2 \mathbb{E}\left[\lim_{\epsilon \to 0} \frac{1}{\epsilon}\left(\varPhi _k+ 2\epsilon \Re\left\{\boldsymbol{\kappa}_{kl}^H \mathbf{v}_{kl}\right\} + o(\epsilon) - \varPhi _k\right)\right],
    \end{aligned}
\end{equation}
where $g\left(\varUpsilon_{ki}\right) = |\varUpsilon  _k|^2 $, and $\varPhi _k=\sum_{m=1}^M \mathbf{v}_{km}^H \mathbf{v}_{km}$. Similar to \eqref{f_gamma_epsilon}, by applying the first-order Taylor expansion to $g(\varUpsilon  _k+\epsilon\boldsymbol{\kappa}_{kl}^H\mathbf{h}_{il})$ and substituting the result into \eqref{delta_con_app}, we obtain the desired result in \eqref{1st_variation_con}.
To this end, we have the first-order variation of the Lagrangian function $\mathcal{L}_k$ w.r.t. $\mathbf{v}_{kl}$ as \eqref{delta_Lk_vkl}.
\begin{figure*}
    \begin{equation}\label{delta_Lk_vkl}
        \frac{\delta \mathcal{L}_k}{\delta \mathbf{v}_{kl}}
        = 2\Re\left\{
            \mathbb{E}\left[
                \boldsymbol{\kappa}_{kl}^H
                \left(
                    \frac{p_k}{\ln 2}\frac{\mathbf{h}_{kl}\left( \sum_{m=1}^M \mathbf{v}_{km}^H \mathbf{h}_{km} \right)^H}{1 + p_k\left| \sum_{m=1}^M \mathbf{v}_{km}^H \mathbf{h}_{km} \right|^2} 
                     + \mu_k \left( \sum_{i=1\backslash k}^K p_i \mathbf{h}_{il} \left( \sum_{m=1}^M \mathbf{v}_{km}^H \mathbf{h}_{im} \right)^H + \sigma^2 \mathbf{v}_{kl} \right)
                \right)
            \right]
        \right\}.
\end{equation}
\noindent\rule[0.25\baselineskip]{\textwidth}{0.8pt}
\end{figure*}
This completes the proof.\hfill$\blacksquare$

\section{Proof of Lemma 1}\label{app:lemma}
{\bf{Proof:}} We further generalize the {du Bois-Raymond} lemma \cite{duBois-Raymond}, i.e., the basic lemma of variational calculus, as follows:

\subsection{Eliminating the Real Part Symbol}
Fix an arbitrary test function $\boldsymbol{\eta}\in\mathcal{H}$ that satisfies
\begin{equation}
    \Re\!\left( \int_{\mathcal{X}} \boldsymbol{\eta}(\mathbf{x})^{H}\,\mathbf f(\mathbf{x})\,d\mu(\mathbf{x}) \right) = 0,
\end{equation}
and we set
\begin{equation}
    z=\int_{\mathcal{X}} \boldsymbol{\eta}(\mathbf{x})^{H}\,\mathbf f(\mathbf{x})\,d\mu(\mathbf{x})\in\mathbb{C}.
\end{equation}
Then $\Re(z)=0$, so we can write $z=i b$ with $b\in\mathbb{R}$.  
Now construct a new test function $\boldsymbol{\eta}'$ defined by
\begin{equation}
    \boldsymbol{\eta}'(\mathbf{x}) = i\,\boldsymbol{\eta}(\mathbf{x}).
\end{equation}

Compute the new inner product:
\begin{equation}
    \begin{aligned}
        &z' = \int_{\mathcal{X}} \boldsymbol{\eta}'(\mathbf{x})^{H}\,\mathbf f(\mathbf{x})\,d\mu(\mathbf{x})
        = \int_{\mathcal{X}} (-i\,\boldsymbol{\eta}(\mathbf{x})^{H})\,\mathbf f(\mathbf{x})\,d\mu(\mathbf{x})\\
        &= -i \int_{\mathcal{X}} \boldsymbol{\eta}(\mathbf{x})^{H}\,\mathbf f(\mathbf{x})\,d\mu(\mathbf{x})
        = -i z
        = -i (i b)
        = b.
    \end{aligned}
\end{equation}
Since $b\in\mathbb{R}$, we have $\Re(z')=b=0$. Therefore $b=0$, and hence $z=0$.

Because $\boldsymbol{\eta}$ was arbitrary, we obtain:
\begin{equation}
    \int_{\mathcal{X}} \boldsymbol{\eta}(\mathbf{x})^{H}\,\mathbf f(\mathbf{x})\,d\mu(\mathbf{x}) = 0,
    \quad \forall\,\boldsymbol{\eta}\in\mathcal{H}.
\end{equation}

\subsection{Basic Lemma of Variational Calculus (du Bois-Raymond lemma \cite{duBois-Raymond})}
Choose $\boldsymbol{\eta}(\mathbf{x})=\mathbf f(\mathbf{x})$ to get
\begin{equation}
    0 = \int_{\mathcal{X}} \mathbf f(\mathbf{x})^{H}\,\mathbf f(\mathbf{x})\,d\mu(\mathbf{x})
      = \int_{\mathcal{X}} \|\mathbf f(\mathbf{x})\|^{2}\,d\mu(\mathbf{x}).
\end{equation}
The integrand $\|\mathbf f(\mathbf{x})\|^{2}\ge 0$ for all $\mathbf{x}$, and its integral is zero. By the basic properties of measure theory we conclude
\begin{equation}
    \|\mathbf f(\mathbf{x})\|^{2}=0 \quad \mu\text{-almost everywhere},
\end{equation}
i.e. $\mathbf f(\mathbf{x})=\mathbf{0}\quad \mu\text{-almost everywhere}.$ This completes the proof.
\hfill$\blacksquare$  

\section{Proof of \eqref{first-order-deltaLk}}\label{app:station_point}
{\bf{Proof:}}
Let us first define the function inside the expectation in \eqref{delta_Lk_vkl} except for the term involving $\boldsymbol{\kappa}_{kl}$ as $\mathbf{F}_{kl}$. Next, to solve for $\frac{\delta \mathcal{L}_k}{\delta \mathbf{v}_{kl}}$, we have
\begin{equation}
    \begin{aligned}
        &\frac{\delta \mathcal{L}_k}{\delta \mathbf{v}_{kl}} = 2\Re\left\{ \mathbb{E}\left[ \boldsymbol{\kappa}_{kl}^H \mathbf{F}_{kl} \right]\right\}\\
        &=\int_{\mathbf{H}_1}..\int_{\mathbf{H}_M} 2\Re\left\{ \boldsymbol{\kappa}_{kl}^H \mathbf{F}_{kl} \right\} \left(\prod_{m=1}^{M}p(\mathbf{H}_m) d\mathbf{H}_m\right)\\
        &=\int_{\mathbf{H}_l} 2\Re\left\{ \boldsymbol{\kappa}_{kl}^H \mathbb{E}\left[ \mathbf{F}_{kl} | \mathbf{H}_l \right] \right\} p(\mathbf{H}_l) d\mathbf{H}_l.\\
    \end{aligned}
\end{equation}
Then, according to \textbf{Lemma \ref{lem_vm}}, we have $\mathbb{E}\left[ \mathbf{F}_{kl} | \mathbf{H}_l \right]=0$.
This completes the proof.\hfill$\blacksquare$

\section{Proof of Lemma 2}\label{app:lemma_appro}
{\bf{Proof:}}
For user $k$ and AP $l$, given any $i \neq k$ and $m \neq l$, let
\begin{equation}
    c_{k,i,m} \triangleq \mathbb{E}\left[ \left( \mathbf{v}_{km}^H \mathbf{h}_{im} \right)^H \right]
    = \mathbb{E}\left[ \mathbf{h}_{im}^H \mathbf{v}_{km} \right].
\end{equation}
Since each AP channel is independent and $\mathbf{h}_{im} \sim \mathcal{CN}(0, \mathbf{R}_{im})$ is a circular symmetric complex Gaussian distribution, we only need to prove that $c_{k,i,m} = 0$. Substitute $c_{k,i,m}$ into the expression of $\mathbf{v}_{kl}$ in \eqref{q-lmmse str}, we have:
\begin{equation}\label{rc_c}
    \begin{aligned}
       \mathbf{v}_{kl}
    =& - \left( \sum_{j \neq k}^K p_j \mathbf{h}_{jl}\mathbf{h}_{jl}^H + \sigma^2 \mathbf{I} \right)^{-1}\enspace\times\\
     &
    \Bigg[
        \mathbf{h}_{kl} \cdot \eta_{kl}(\mathbf{H}_l)
        + \sum_{r \neq k}^K p_r \mathbf{h}_{rl} \cdot \sum_{m \neq l}^L c_{k,r,m}
    \Bigg], 
    \end{aligned}
\end{equation}
where $\eta_{kl}(\mathbf{H}_l)\triangleq \mathbb{E}[
        \frac{p_k}{\mu_k \ln 2}\frac{( \sum_{m=1}^M \mathbf{v}_{km}^H \mathbf{h}_{km} )^H}{1 + p_k| \sum_{m=1}^M \mathbf{v}_{km}^H \mathbf{h}_{km} |^2}\,|\, \mathbf{H}_l]$. Then, multiply $\mathbf{v}_{km}$ by $\mathbf{h}_{im}^H$ from both sides and take the expectation:
\begin{equation}
    \begin{aligned}
    &c_{k,i,m}
    = \mathbb{E}\left[ \mathbf{h}_{im}^H \mathbf{v}_{km} \right] \\
    &= -\mathbb{E}\Bigg[ \mathbf{h}_{im}^H \left( \sum_{j=1\backslash k}^K p_j \mathbf{h}_{jm}\mathbf{h}_{jm}^H + \sigma^2 \mathbf{I} \right)^{-1}
        \mathbf{h}_{km} \cdot \eta_{km}(\mathbf{H}_m) \Bigg] \\
    & -\sum_{r \neq k}^K  \sum_{t \neq m}^M p_r c_{k,r,t} \cdot \mathbb{E}\Bigg[ \mathbf{h}_{im}^H \left( \sum_{j \neq k} p_j \mathbf{h}_{jm}\mathbf{h}_{jm}^H + \sigma^2 \mathbf{I} \right)^{-1} \mathbf{h}_{rm} \Bigg].
    \end{aligned}
\end{equation}
Considering the first term on the right-hand side, we first define $\mathbf{W}_{-ki,m} \triangleq \sum_{j \neq k, i} p_j \mathbf{h}_{jm}\mathbf{h}_{jm}^H + \sigma^2 \mathbf{I}$, and then obtain
\begin{equation}
    \sum_{j=1\backslash k}^K p_j \mathbf{h}_{jm}\mathbf{h}_{jm}^H + \sigma^2 \mathbf{I}
    = \mathbf{W}_{-ki,m} + p_i \mathbf{h}_{im}\mathbf{h}_{im}^H.
\end{equation}
Afterward, using the Sherman-Morrison formula {\cite{Horn2013CUP}}, we have
\begin{equation}\label{sm_1}
    \begin{aligned}
        &\left( \mathbf{W}_{-ki,m} + p_i \mathbf{h}_{im}\mathbf{h}_{im}^H \right)^{-1}\\
   & = \mathbf{W}_{-ki,m}^{-1} - \frac{p_i \mathbf{W}_{-ki,m}^{-1}\mathbf{h}_{im}\mathbf{h}_{im}^H\mathbf{W}_{-ki,m}^{-1}}{1 + p_i \mathbf{h}_{im}^H\mathbf{W}_{-ki,m}^{-1}\mathbf{h}_{im}}.    
    \end{aligned}
\end{equation}
Then, we can further derive the term $\mathbf{h}_{im}^H ( \mathbf{W}_{-ki,m} + p_i \mathbf{h}_{im}\mathbf{h}_{im}^H)^{-1} \mathbf{h}_{km}$ in the first term as
\begin{equation}
\begin{aligned}
& \mathbf{h}_{im}^H \mathbf{W}_{-ki,m}^{-1} \mathbf{h}_{km}
    - \frac{p_i \mathbf{h}_{im}^H \mathbf{W}_{-ki,m}^{-1}\mathbf{h}_{im}\mathbf{h}_{im}^H\mathbf{W}_{-ki,m}^{-1} \mathbf{h}_{km}}{1 + p_i \mathbf{h}_{im}^H\mathbf{W}_{-ki,m}^{-1}\mathbf{h}_{im}} \\
    &= \frac{\mathbf{h}_{im}^H \mathbf{W}_{-ki,m}^{-1} \mathbf{h}_{km}}{1 + p_i \mathbf{h}_{im}^H\mathbf{W}_{-ki,m}^{-1}\mathbf{h}_{im}}.
    \end{aligned}
\end{equation}
Since $\mathbf{W}_{-ki,m}$ is independent of $\mathbf{h}_{im}$ and $\mathbf{h}_{km}$, $\mathbf{W}_{-ki,m}^{-1}$ can be regarded as a deterministic matrix given $\mathbf{H}_m \setminus \{\mathbf{h}_{im}, \mathbf{h}_{km}\}$. Under this condition, the integrand $\frac{\mathbf{h}_{im}^H \mathbf{W}_{-ki,m}^{-1} \mathbf{h}_{km}}{1 + p_i \mathbf{h}_{im}^H\mathbf{W}_{-ki,m}^{-1}\mathbf{h}_{im}}$ has odd symmetry with respect to $\mathbf{h}_{im}$: the numerator $\mathbf{h}_{im}^H \mathbf{W}_{-ki,m}^{-1} \mathbf{h}_{km}$ is a linear function (odd function) of $\mathbf{h}_{im}$, while the denominator $1 + p_i \mathbf{h}_{im}^H\mathbf{W}_{-ki,m}^{-1}\mathbf{h}_{im}$ is a quadratic form (even function) of $\mathbf{h}_{im}$, making the overall function an odd function. Since $\mathbf{h}_{im} \sim \mathcal{CN}(0, \mathbf{R}_{im})$ has a distribution symmetric about the origin (i.e., $\mathbf{h}_{im}$ and $-\mathbf{h}_{im}$ follow the same distribution), the expectation of an odd function over a symmetric distribution is zero; therefore, the first term is zero.

For the second term, for the case when $r \neq i$, similarly using the Sherman-Morrison formula, it can be known that the expectation is zero. For the case when $r = i$, we define
\begin{equation}
    \beta_{k,i,m} \triangleq \mathbb{E}\Bigg[ \mathbf{h}_{im}^H \left( \sum_{j=1\backslash k}^K p_j \mathbf{h}_{jm}\mathbf{h}_{jm}^H + \sigma^2 \mathbf{I} \right)^{-1} \mathbf{h}_{im} \Bigg].
\end{equation}
Since the inverse matrix is positive definite and $\mathbf{h}_{im} \neq 0$, we have $\beta_{k,i,m} > 0$. Therefore, we have
\begin{equation}\label{ckim}
    c_{k,i,m} = -p_i  \beta_{k,i,m}\sum\nolimits_{t \neq m}^M c_{k,i,t}.
\end{equation}
To solve the equation system, we first analyze the range of values for $\beta_{k,i,m}$. Using a Sherman-Morrison expansion similar to that in \eqref{sm_1}, we separate $p_i \mathbf{h}_{im}\mathbf{h}_{im}^H$ from the matrix to be inverted (note that the summation term $\sum_{j \neq k}$ includes the case $j=i$):
\begin{equation}    
\mathbf{h}_{im}^H \left( \mathbf{W}_{-ki,m} + p_i \mathbf{h}_{im}\mathbf{h}_{im}^H \right)^{-1} \mathbf{h}_{im}    
= \frac{\mathbf{h}_{im}^H \mathbf{W}_{-ki,m}^{-1} \mathbf{h}_{im}}{1 + p_i \mathbf{h}_{im}^H \mathbf{W}_{-ki,m}^{-1} \mathbf{h}_{im}}.
\end{equation}
Since $\mathbf{W}_{-ki,m} \succ 0$, let $x_{k,i,m} = \mathbf{h}_{i,m}^H \mathbf{W}_{-ki,m}^{-1} \mathbf{h}_{i,m} \geq 0$, then the right side of the equation is $\frac{x_{k,i,m}}{1 + p_i x_{k,i,m}}$. Clearly, for any $x_{k,i,m} \geq 0$ and $p_i > 0$, we have
\begin{equation}
0 \leq \frac{x_{k,i,m}}{1 + p_i x_{k,i,m}} < \frac{1}{p_i}.
\end{equation}
Taking expectation on both sides, we have $0 \leq \beta_{k,i,m} < \frac{1}{p_i}$, which implies $0 \leq p_i \beta_{k,i,m} < 1$. Now let $S_{k,i} \triangleq \sum_{t=1}^M c_{k,i,t}$, then $\sum_{t \neq m}^M c_{k,i,t} = S_{k,i} - c_{k,i,m}$. Substituting this into \eqref{ckim} for $c_{k,i,m}$, we have
\begin{subequations}
    \begin{align}
    &c_{k,i,m} = -p_i \beta_{k,i,m} (S_{k,i} - c_{k,i,m}),\\
    &\Rightarrow c_{k,i,m} = \frac{-p_i \beta_{k,i,m}}{1 - p_i \beta_{k,i,m}} S_{k,i}.
\end{align}
\end{subequations}
Summing over all $c_{k,i,m}$ for all $m\in\{1, \dots,M\}$:
\begin{equation}
    S_{k,i} = \sum_{l=1}^M c_{k,i,m} = S_{k,i} \sum_{m=1}^M \frac{-p_i \beta_{k,i,m}}{1 - p_i \beta_{k,i,m}}.
\end{equation}
Rearranging, we obtain
\begin{equation}
    S_{k,i} \left( 1 + \sum_{m=1}^M \frac{p_i \beta_{k,i,m}}{1 - p_i \beta_{k,i,m}} \right) = 0.
\end{equation}
Since $0 \leq p_i \beta_{k,i,m} < 1$, the term in parentheses is strictly greater than $1$. Therefore, the only possible solution is $S_{k,i} = 0$. Then, by the relationship between $c_{k,i,m}$ and $S_{k,i}$, it follows that $c_{k,i,m} = 0, \forall m$.

In summary, for any $i \neq k$ and $m \neq l$, we have $\mathbb{E}\left[ \mathbf{h}{im}^H \mathbf{v}{km} \right] = 0$. This completes the proof.\hfill$\blacksquare$

\end{appendices}

\bibliographystyle{IEEEtran}

\bibliography{ref}

\end{document}